\documentclass[a4paper,11pt]{article}

\usepackage[margin=2.5cm]{geometry}
\usepackage[utf8]{inputenc}

\usepackage{amsmath}
\usepackage{amsfonts}
\usepackage{amssymb}
\usepackage{lipsum}
\usepackage{amsthm}
\usepackage{hyperref}

\usepackage[sort&compress,nameinlink,noabbrev,capitalize]{cleveref}
\crefname{algorithm}{Algorithm}{Algorithms}

\usepackage{booktabs}
\usepackage{tabularx}
\usepackage{multirow}

\usepackage[citestyle=numeric,
            bibstyle=numeric,
            hyperref=true,
            url=false,
            isbn=false,
            backref=true,
            backrefstyle=three,
            maxcitenames=3,
            maxbibnames=100,
            block=none,
            bibencoding=utf8,
            backend=biber]{biblatex}

\bibliography{literature} %

\usepackage{graphicx}
\graphicspath{{./graphics/}}%

\theoremstyle{plain}
\newtheorem{theorem}{Theorem}[section]
\newtheorem{lemma}[theorem]{Lemma}

\newtheorem{observation}[theorem]{Observation}
\newtheorem{proposition}[theorem]{Proposition}

\theoremstyle{definition}
\newtheorem{definition}[theorem]{Definition}

\newcommand{\proofparagraph}[1]{\smallskip\emph{#1}}

\usepackage{authblk}

\title{Turbocharging Heuristics for Weak Coloring Numbers\footnote{This paper is based on Alexander Dobler's master thesis \cite{dobler2021turbocharging}.
    Alexander Dobler was supported by the Vienna Science and Technology Fund (WWTF) under grant ICT19-035.
    Manuel Sorge was supported by the Alexander von Humboldt Foundation.
    Anaïs Villedieu was supported by the Austrian Science Fund (FWF) under grant~P31119. }}

\author{Alexander Dobler}
\author{Manuel Sorge}
\author{Anaïs Villedieu}

\affil{TU Wien, Austria, $\{\texttt{adobler}, \texttt{manuel.sorge}, \texttt{avilledieu}\}$\texttt{@ac.tuwien.ac.at}}

\usepackage{scrextend}          %
\usepackage{needspace}          %
\usepackage[ruled,linesnumbered]{algorithm2e} %
\usepackage[textsize=footnotesize]{todonotes}
\usepackage{xcolor,colortbl}

\newcommand{\poly}{\ensuremath{\operatorname{poly}}}
\newcommand{\wcol}{\ensuremath{\mathrm{wcol}}}
\newcommand{\wrc}{\ensuremath{\wcol_r}}       %
\newcommand{\wreach}{\ensuremath{\mathrm{Wreach}}}
\newcommand{\N}{\ensuremath{\mathbb{N}}}
\newcommand{\wreachable}{weakly $r$\nobreakdash-reachable}
\newcommand{\wrcol}{weak $r$\nobreakdash-coloring number}
\newcommand{\icwrcolleftlong}{\textup{\textsc{Incremental Conservative Weak $r$\nobreakdash-coloring}}}
\newcommand{\icwrcolleft}{\textup{\textsc{IC-WCOL($r$)}}}
\newcommand{\islong}{\textup{\textsc{Independent Set}}}

\newcommand{\wcolmerge}{\textup{\textsc{WCOL-Merge($r$)}}}

\newcommand{\placebefore}{\mathrm{placebefore}}
\newcommand{\placeafter}{\mathrm{placeafter}}
\newcommand{\dist}{\ensuremath{\mathrm{dist}}}

\newcommand{\degen}{\ensuremath{\mathrm{degeneracy}}}
\newcommand{\wrcolprob}{\textup{\textsc{WCOL($r$)}}}

\newcommand{\wrcolproblong}{\textup{\textsc{Weak $r$\nobreakdash-coloring Number}}}
\newcommand{\TCLASTC}{\icwrcolleft}
\newcommand{\TCMERGE}{\wcolmerge}
\newcommand{\TCLASTCRL}{\textup{\textsc{IC-WCOL-RL($r$)}}}

\newcommand{\prob}[6]{%
  \needspace{3\baselineskip}
  \begin{quote}
    \begin{labeling}{#6}%
    \item[#1]
    \item[\emph{#2}]#3
    \item[\emph{#4}]#5
    \end{labeling}%
  \end{quote}%
}

\newcommand{\probdef}[3]{\prob{#1}{Instance:}{#2}{Question:}{#3}{as}}

\begin{document}

\maketitle
\thispagestyle{empty}
\begin{abstract}
  \looseness=-1
  Bounded expansion and nowhere-dense classes of graphs capture the theoretical tractability for several important algorithmic problems.
  These classes of graphs can be characterized by the so-called weak coloring numbers of graphs, which generalize the well-known graph invariant degeneracy (also called $k$-core number).
  Being NP-hard, weak-coloring numbers were previously computed on real-world graphs mainly via incremental heuristics.
  We study whether it is feasible to augment such heuristics with exponential-time subprocedures that kick in when a desired upper bound on the weak coloring number is breached.
  We provide hardness and tractability results on the corresponding computational subproblems.
  We implemented several of the resulting algorithms and show them to be competitive with previous approaches on a previously studied set of benchmark instances containing 86 graphs with up to 183831 edges.
  We obtain improved weak coloring numbers for over half of the instances.
\end{abstract}

\pagenumbering{arabic}

\section{Introduction}\label{sec:introduction}
A \emph{degeneracy ordering} of a graph~$G$ can be obtained by iteratively removing an arbitrary vertex of minimum degree from~$G$ and putting it in front of the current ordering~\cite{DBLP:journals/jacm/MatulaB83}.
The \emph{degeneracy} of a graph is the largest degree of a vertex encountered at removal.
Degeneracy orderings are immensely useful when solving various tasks on graphs both in theory and practice~\cite{cai_random_2006,eppstein_listing_2013,komusiewicz_algorithmic_2015,pinar_escape_2017,himmel_adapting_2017,bressan_faster_2021}.
A key observation is that many graphs in practice have small degeneracy (e.g., in the order of hundreds for millions of edges)~\cite{eppstein_listing_2013}.
Thus, when looking for, e.g., maximum-size cliques, it is sufficient to look within the small number of neighbors in front of each vertex in the degeneracy ordering.

However, degeneracy is not robust under local changes: E.g., contracting a set of disjoint stars in a graph may arbitrarily increase its degeneracy.
This property makes problems intractable on graphs that have small degeneracy if these problems are less local than finding maximum-size cliques.
For example, detecting mild clique relaxations is hard on graphs of bounded degeneracy~\cite{komusiewicz_algorithmic_2015,huffner_finding_2015}.
Hence, we are searching for robust sparsity measures within the framework of structural sparsity~\cite{nesetril_sparsity_2012}.

We can obtain a robust variant of the class of graphs with bounded degeneracy by using the family of measures called weak $r$\nobreakdash-coloring numbers~\cite{DBLP:journals/order/KiersteadY03}.
For an integer~$r$, the weak $r$\nobreakdash-coloring number (\wrc) of a graph $G$ is the least integer $k$ such that there is a vertex ordering with the property that, for each vertex $v$, there are at most~$k$ vertices~$u$ that are reachable from $v$ by a path~$P$ of length at most~$r$ such that $P$ does not use vertices that come before~$u$ in the ordering.
The integer~$r$, also called radius, interpolates between the degeneracy plus one ($r = 1$) and the so-called treedepth of~$G$ ($r = |V(G)|$~\cite{DBLP:journals/siamdm/GroheKRSS18}), a measure of tree-likeness~\cite{alexpothen_complexity_1988,bodlaender_rankings_1998,katchalski_ordered_1995} and target of a recent implementation challenge~\cite{kowalik_pace_2020}.

\looseness=-1
It is important to compute the \wrc\ of real-world graphs for two reasons.
First, the \wrc\ plays an important role in algorithmic and combinatorial techniques in structural sparsity~(e.g., \cite{akhoondianamiri_distributed_2018,brianski_erdos_2021,dreier_lacon_2021,dvorak_constantfactor_2013,eickmeyer_neighborhood_2017,joret_nowhere_2019,kwon_low_2020,nesetril_clustering_2020,pilipczuk_parameterized_2018,reidl_characterising_2019}).
For instance, a central concept therein is nowhere denseness and a subgraph-closed class of graphs is nowhere dense if and only if for each fixed integer $r$ and $\epsilon > 0$ each graph $G$ in the class has \wrc\ at most~$O(|V(G)|^{\epsilon})$~\cite{siebertz_nowhere_2019a}.
Thus, obtaining the \wrc s of real-world graphs will help us gauge how well this theory fits practice.
Second, there is a prospect that \wrc s will help us solve other computational problems on real-world graphs more efficiently:
On nowhere-dense graph classes each problem expressible in first-order logic can be solved in near-linear time~\cite{grohe_deciding_2017}.
There is indeed indication that relevant classes of real-world networks, including certain scale-free networks, are nowhere dense~\cite{demaine_structural_2019}.
Thus, \wrc s may help us transfer the above theoretical near-linear-time algorithms into something practically useful.
For example, this work is underway for counting subgraphs~\cite{obrien_experimental_2017,reidl_coloravoiding_2020}, which is the underlying computational problem of computing graph motifs or graphlets in biological and social networks~\cite{milo_network_2002,przulj_modeling_2004}.

Computing the \wrc\ of a graph is an algorithmically challenging task: for each $r \geq 2$ it is NP-complete~\cite{DBLP:journals/siamdm/GroheKRSS18,DBLP:journals/corr/abs-2112-10562}.
So far, there is but one work that studies computing upper bounds for \wrc s in real-world graphs:
Nadara et al.~\cite{DBLP:journals/jea/NadaraPRRS19} used greedy heuristics that build the associated vertex ordering by iteratively choosing a yet unordered vertex that seems favorable and putting it at the front (or the back) of the current subordering (an ordering of a subset of all vertices).
Afterwards, they apply local-search techniques that make local shifts in the ordering that decrease the associated weak $r$-coloring number.
This yields upper bounds; to date the true weak $r$-coloring numbers of the studied graphs are unknown, even for the smallest part of Nadara et al.'s dataset that contains graphs with 62 to 930 edges.

\looseness=-1
In this work, we study a paradigm that has previously been successfully applied to improve the quality of the results computed by greedy heuristics: turbocharging~\cite{sepphartung_incremental_2013,downey_dynamic_2014,DBLP:conf/tamc/Abu-KhzamCESW17,DBLP:journals/algorithmica/GaspersGJMR19}.
The basic idea is that we pre-specify an upper bound~$k \in \mathbb{N}$ on the \wrc\ of the ordering that we want to compute.
We carry out the greedy heuristic that computes the ordering iteratively.
If at some point it can be detected that the ordering will yield \wrc\ larger than~$k$ --- the \emph{point of regret} --- we start a turbocharging algorithm.
This algorithm tries to modify the current ordering, by reordering or replacing certain vertices, so as to make it possible to achieve \wrc\ at most~$k$ again.
Then we continue with the greedy heuristic.

\looseness=-1
Our contribution is to develop the turbocharging algorithms applied at the point of regret, obtaining running-time guarantees and lower bounds, and to implement, engineer, and evaluate these algorithms on Nadara et al.'s dataset.
The two main approaches that we study are as follows.
We fix a \emph{reconstruction parameter} $c \in \mathbb{N}$.
At the point of regret, in order to obtain a subordering of \wrc\ at most $k$ we either (a) replace the last $c$ vertices of the current subordering with new, yet unordered vertices or (b) take $c$ vertices out of the ordering and merge them into the subordering at possibly different positions.
The formal definitions are given in \cref{sec:preliminaries}.
We show that both approaches are NP-hard in general (see \cref{sec:complexity}) and hence we also consider the influence of small parameters on the achievable running-time guarantees.
That is, we aim to show \emph{fixed-parameter tractability} by giving algorithms with $f(p) \cdot n^{O(1)}$ running time for a small parameter $p$ and input size $n$.
On the negative side, approach (a) is W[1]-hard with respect to even both~$c$ and~$k$ (\cref{thm:hardicwrcolleft}).
That is, an algorithm with running time $f(c, k) \cdot \poly(n)$ is unlikely, where $n$ is the number of vertices.
This stands in contrast to Gaspers et al.~\cite{DBLP:journals/algorithmica/GaspersGJMR19} who obtained such an algorithm when the goal is to compute the treewidth of the input graph instead of its \wrc.
On the positive side, approach (a) is trivially tractable in polynomial time for constant~$c$.
For approach~(b), we obtain a fixed-parameter algorithm with respect to $c + k$ (\cref{thm:mergefpt}).
We implemented a set of algorithms including the two previously mentioned ones and report on implementation considerations and results in \cref{sec:impexp}.
The results indicate that on average the weak $r$-coloring numbers achieved by Nadara et al.~\cite{DBLP:journals/jea/NadaraPRRS19} can be improved by~5\,\% by using our turbocharging algorithms.
Using turbocharging we obtain smaller weak coloring numbers than all previous approaches on 181 of the in total 334 instances used by Nadara et al.~\cite{DBLP:journals/jea/NadaraPRRS19}.
We also establish a baseline for lower bounds on the coloring numbers in \cref{sec:lower-bound}.

\section{Preliminaries and turbocharging problems}\label{sec:preliminaries}
\subparagraph*{General preliminaries.}
\begin{sloppypar}
We only consider undirected, unweighted graphs $G$ without loops.
By~$V(G)$ and~$E(G)$ we denote the vertex set and edge set of $G$, respectively.
For~$S\subseteq V(G)$,~$G[S]$ is the \emph{induced subgraph} on vertices in~$S$.
A \emph{path}~$P=(v_1,\dots,v_n)$ in~$G$ is a non-empty sequence of vertices, such that consecutive vertices are connected by an edge.
The \emph{length} of a path is~$|V(P)|-1$.
In particular, a path of length~$0$ consists of a single vertex.
We use $\dist_G(u,v)$ to denote the length of the shortest path between vertices~$u$ and~$v$ in~$G$.
\end{sloppypar}

A \emph{vertex ordering}~$L$ of a graph $G$ is a linear ordering of $V(G)$.
We write $u\prec_L v$ if vertex~$u$ precedes vertex~$v$ in~$L$.
In this case we say that~$u$ is \emph{left of}~$v$ w.r.t.~$L$.
Equivalently, we write $u\preceq v$ if $u\prec v$ or $u=v$.
We also denote vertex orderings~$L$ as sequences of its elements, that means $L=(v_1,\dots, v_n)$ represents the vertex ordering where $v_i\prec_L v_j$ iff.\ $i<j$.
We denote by $\Pi(G)$ the set of all vertex orderings of $G$.
A \emph{subordering} is a linear ordering of a subset $S \subseteq V(G)$.
The notation $L_S$ shall always denote a subordering where $S$ is the set of vertices ordered in the subordering.
We call the vertices in \(V(G) \setminus S\) \emph{free} with respect to \(L_S\).
Usually we will denote the set of free vertices with respect to a subordering by \(T\).
For a subordering~$L_S$ and $S'\subseteq S$ we denote by~$L_{S}[S^\prime]$ the subordering that \emph{agrees with}~$L_S$ on~$S^\prime$, that is, for all $u,v\in S^\prime$ we have that $u\prec_{L_S[S^\prime]}v$ iff.\ $u\prec_{L_S} v$, and all vertices in $V(G) \setminus S'$ are free w.r.t.~$L_S[S']$.
For a subordering~$L_S$ and $S^\prime\supseteq S$, a subordering~$L_{S^\prime}$ is a \emph{right extension} of~$L_S$ if $L_{S^\prime}[S]=L_S$ and~$u\prec_{L_{S^\prime}}v$ for all $u\in S$ and $v\in S^\prime\setminus S$.
If~$S^\prime=V$, then~$L_{S^\prime}$ is called \emph{full right extension}.

\subparagraph*{Weak coloring numbers.}
For a vertex ordering~$L$ of~$G$ and $r\in\N$, a vertex~$u$ is \emph{\wreachable} from a vertex~$v$ w.r.t.~$L$ if there exists a path~$P$ of length $\ell$ with $0 \le \ell \le r$ between~$u$ and~$v$ such that $u\preceq_L w$ for all $w\in V(P)$.
Let $\wreach_r(G,L,v)$ be the set of vertices that are \wreachable\ from~$v$ in~$G$ w.r.t.~$L$.
The \emph{\wrcol} $\wcol_r(G,L)$ of a vertex ordering~$L$ is $\wcol_r(G,L)=\max_{v\in V(G)}|\wreach_r(G,L,v)|$,
and the weak $r$-coloring number $\wcol_r(G)$ of~$G$ is $\wcol_r(G)=\min_{L\in\Pi}\wcol_r(G,L)$.
The main decision problem related to \wrcol\ is the following.
\probdef{\wrcolproblong\ (\wrcolprob)}
{A graph $G=(V,E)$, and an integer~$k$.}
{Does~$G$ have weak $r$\nobreakdash-coloring number at most~$k$?}

\subparagraph*{Turbocharging.}
Our goal is to find a vertex ordering~$L$ of a given graph~$G$ with small $\wcol_r(G,L)$ by applying turbocharging.
We start with two well-known iterative greedy heuristics (descriptions will follow) of Nadara et al.\ \cite{DBLP:journals/jea/NadaraPRRS19} that build vertex orderings with small \wrcol\ from left to right.
That is, these heuristics start with the empty subordering~$L_S=\emptyset$ and in each step compute a right extension of~$L_S$ that contains one more vertex.
This process is continued until the constructed subordering contains all vertices.
A key observation about this process that we can use for turbocharging is that the size of the weakly reachable set of each vertex cannot decrease:
\begin{observation}\label{obs:wreach-grows}
  Let~$L_S$ be a subordering, $u,v\in V(G)$, and~$L$ be a full right extension of~$L_S$.
  If
  $u=v$, or
  $u\in S$ and there exists a path $P$ of length $\ell$ with $0\le \ell\le r$ between~$u$ and~$v$ such that
  $u\preceq_{L_S} w$ for all $w\in V(P)\cap S$,
  then $u\in\wreach_r(G,L,v)$.
\end{observation}
For~$u$ and~$v$ as in \cref{obs:wreach-grows} we extend the definition of weak $r$\nobreakdash-reachability to suborderings~$L_S$ by defining %
$u\in\wreach_r(G,L_S,v)$ and
$\wcol_r(G,L_S)=\max_{v\in V(G)}|\wreach_r(G,L_S,v)|$.
We immediately obtain that $\wcol_r(G,L_S)$ is a lower bound; that is, if~$L$ is a full right extension of~$L_S$ (such as obtained by one of the heuristics), then $\wcol_r(G,L)\ge \wcol_r(G,L_S)$.

The two heuristics of Nadara et al.\ that we apply are:
\begin{itemize}
  \item The Degree-Heuristic: This heuristic orders vertices by descending degree, ties are broken arbitrarily.
  \item The Wreach-Heuristic: For a subordering~$L_S$, this heuristic picks the free vertex~$v$ with the largest $\wreach_r(G,L_S,v)$.
        Ties are broken by descending degree.
\end{itemize}
\looseness=-1
Nadara et al.\ proposed several other heuristics, but those heuristics do not build vertex orderings from left to right, but in different orders.
Additionally, the above heuristics are among the best-performing ones with regard to computed weak coloring numbers and runtime.

\looseness=-1
In what follows, let us assume that we want to compute a vertex ordering~$L$ of a graph~$G$ with $\wcol_r(G,L)\le k$ where $k\in\N$.
We might apply one of the heuristics until we obtain a subordering~$L_S$ such that $\wcol_r(G,L_S)>k$.
We call such a subordering \emph{non-extendable} (and otherwise the subordering is \emph{extendable}); we also say that this point in the execution of the heuristic is the \emph{point of regret}.
We then consider two exact \emph{turbocharging problems} that try to locally augment~$L_S$, such that it is extendable again.
If the \emph{turbocharging algorithms} for these problems that we later propose are successful in making~$L_S$ extendable, then we continue applying the heuristic until we have to repeat this process (trying to find a vertex order~$L$ with $\wcol_r(G,L)\le k$).

Motivated by a turbocharging algorithm for computing tree-decompositions by Gaspers et al.~\cite{DBLP:journals/algorithmica/GaspersGJMR19} we consider replacing a bounded-length suffix of the current subordering.
That is, we specify a \emph{reconstruction parameter}~$c \in \mathbb{N}$ in advance and, at the point of regret, we remove the last $c$ vertices from $L_S$ and then try to add $c$ (possibly) different vertices.
This leads to the following turbocharging problem.
\probdef{\icwrcolleftlong\ (\icwrcolleft)}
{A graph~$G$, a subordering~$L_S$, and positive integers~$k$ and~$c$.}
{Is there an extendable right extension~$L_{S^\prime}$ of $L_S$ such that $|S^\prime\setminus S|=c$?}
Our second turbocharging approach is based on a vertex~$v$ with $|\wreach_r(G,L_S,v)|>k$.
Therein, instead of the suffix of the current order, we choose a set~$S_2$ of vertices related to the weakly $r$\nobreakdash-reachable set of $v$ (details follow in \cref{sec:impexp}).
We remove the vertices in~$S_2$ from~$L_S$, leaving us with the subordering~$L_{S_1}$, and then we try to reinsert the vertices in~$S_2$ while decreasing the weak coloring number.
\probdef{\wcolmerge}
{A graph~$G$, an integer~$k$, two disjoint sets~$S_1$ and~$S_2$ such that $S_1,S_2\subseteq V(G)$, and a subordering~$L_{S_1}$.}
{Is there an extendable subordering~$L_{S_1\cup S_2}$ such that $L_{S_1\cup S_2}[S_1]=L_{S_1}$?}
Herein, we put the \emph{reconstruction parameter $c$} equal to $|S_2|$ and denote by $T$ the set of free vertices $V(G)\setminus (S_1\cup S_2)$.

\section{Algorithms and running-time bounds}\label{sec:complexity}
We continue by providing algorithmic upper and lower bounds for \icwrcolleftlong\ (\icwrcolleft) and \wcolmerge, starting with \icwrcolleft. %
\subsection{IC-WCOL(\texorpdfstring{\boldmath$r$}))}
The first theoretical result that we want to present is the NP-hardness and W[1]-hardness of \icwrcolleftlong\ for each $r\ge 1$ by giving a reduction from \islong.
\islong\ takes as input a graph~$G$ and a positive integer $p$ and asks if there is a set of vertices~$I$ of size at least $p$ such that $\{u,v\}\not\subseteq I$ for all $\{u,v\}\in E(G)$.
The parameter $p$ --- the desired independent set size of a given \islong\ instance --- is transformed to the parameters $c=p$ and $k=2$ of \icwrcolleft, giving us a parameterized reduction from \islong\ to \icwrcolleft.

\begin{theorem}
  \label{thm:hardicwrcolleft}
  For each fixed $r\ge 1$, \icwrcolleftlong\ is NP-hard and $W[1]$-hard when parameterized by $k+c$.
\end{theorem}
\begin{figure}[t!]
  \centering
  \includegraphics[]{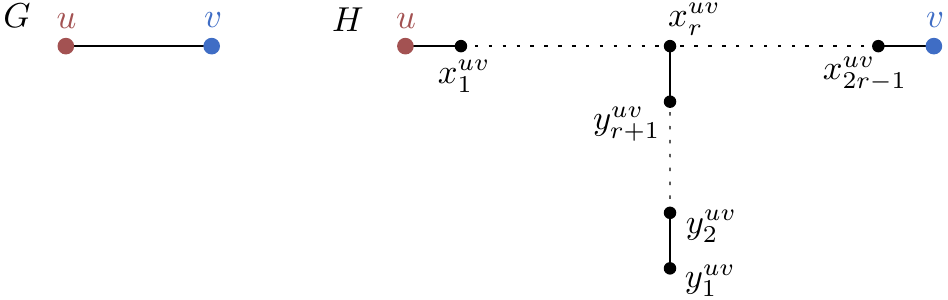}
  \caption{Sketch for the subdivision of edge $\{u,v\}$.}%
  \label{fig:npwrcolleft}
\end{figure}
\begin{proof}
  We give a polynomial and parameterized reduction from \islong\ which is NP-complete \cite{DBLP:books/fm/GareyJ79} and $W[1]$-hard when parameterized by $p$ \cite{DBLP:journals/tcs/DowneyF95} --- the desired independent set.
  Let~$(G,p)$ be an instance of \islong\ ($p$ is the size of the desired independent set).
  We construct an instance $(H,L_S,k,c)$ of \icwrcolleft\ as follows.
  First, let $k=2$ and $c=p$.
  Constructing~$H$ from~$G$ proceeds in two steps.
  \begin{itemize}
    \item We replace each edge~$\{u,v\}$ by a path $(u, x_1^{uv},\dots,x_{2r-1}^{uv}, v)$.
    \item For each edge $\{u,v\}\in E(G)$ we introduce a path $P=(y_{1}^{uv},\dots, y_{r+1}^{uv})$ of~$r+1$ vertices connecting the end vertex $y_{r+1}^{uv}$ to $x_r^{uv}$ (which is the middle vertex on the path between~$u$ and~$v$).  %
  \end{itemize}
  A sketch of the construction for a single edge $\{u,v\}$ is shown in \cref{fig:npwrcolleft}.
  Lastly, we define the subordering~$L_S$ such that~$L_S$ is an arbitrary subordering of vertices $S=\{y_1^{uv}:\{u,v\}\in E(G)\}$.
  We proceed by showing that there is an independent set $I$ of size $p$ if and only if there is an extendable right extension $L_{S^\prime}$ of~$L_S$ with $|S^\prime\setminus S|=c$.

  ``$\Rightarrow$'': Assume that~$G$ has an independent set~$I$ of size $|I|=p=c$.
  Consider an arbitrary right extension $L_{S^\prime}$ of~$L_S$ with $S^\prime\setminus S=I$.
  Clearly, $|S^\prime\setminus S|=c$. We claim that $L_{S'}$ is also extendable.
  To see this, consider $|\wreach_r(H,L_{S^\prime},v)|$ for each vertex in $V(H)$.
  \begin{itemize}
    \item Clearly $\wreach_r(H,L_{S^\prime},v)\subseteq \{v\}$ if $v\in V(G)$, as~$v$ can only weakly $r$\nobreakdash-reach itself in $I$ due to the subdivision of edges.
          Thus, $|\wreach_r(H,L_{S^\prime},v)|\le k$.
    \item As $I$ is an independent set, we also have that $|\wreach_r(H,L_{S^\prime},x_{i}^{uv})|\le 2$ for all $\{u,v\}\in E(G)$ and $1\le i\le 2r-1$.
          This is true because~$u$ and~$v$ are the only potentially weakly $r$\nobreakdash-reachable vertices in $I$, but~$u$ and~$v$ cannot both be in $I$.
          Again $|\wreach_r(H,L_{S^\prime},x^{uv}_i)|\le k$ holds.
    \item Finally, $\wreach_r(H, L_{S^\prime},y^{uv}_j)=\{y_1^{uv},y_j^{uv}\}$ for all $\{u,v\}\in E(G)$ and $1\le j\le r+1$ as all paths to vertices in $I\subseteq V(G)$ have length greater than $r$.
  \end{itemize}
  Summarizing, we have that $\wcol_r(G,L_S)\le k$, as required.

  ``$\Leftarrow$'': Assume that we have found an extendable right extension $L_{S^\prime}$ of~$L_S$ with $|S^\prime\setminus S|=c$.
  We observe that vertices $y^{uv}_i$ with $2\le i\le r+1$ cannot be in $S^\prime\setminus S$ because this would result in $|\wreach_r(H,L_{S^\prime},y^{uv}_j)| > k$ for some $2\le j\le r+1$.
  Consequently, $S^\prime\setminus S$ cannot contain any vertex $x_i^{uv}$, as otherwise $y^{uv}_{r+1}$ could weakly $r$-reach one of the vertices $x_i^{uv}$ appearing left-most in $L_{S^\prime}$.
  Putting this together we have $(S^\prime\setminus S)\subseteq V(G)$. We define $I=(S^\prime\setminus S)$.
  It remains to show that $I$ is an independent set, meaning that we have $\{u,v\}\not\subseteq S$ for edges $\{u,v\}\in E(G)$.
  This has to hold as otherwise $x_r^{uv}$ could weakly $r$-reach both~$u$ and~$v$, meaning that $|\wreach_r(H,L_{S^\prime},x_r^{uv})|\ge k$, which leads to a contradiction.
  Hence, $I$ is an independent set in~$G$ with size~$p$, as required.

  Note that this reduction is polynomial under the assumption that~$r$ is a constant.
  Consequently, \icwrcolleftlong\ is NP-hard.
  Furthermore, as $k=2$ and $c=p$,  we also obtain the stated $W[1]$-hardness.
\end{proof}
We want to note that we can state a similar reduction for every $k>2$.
Namely, for $k>2$ we add $k-2$ vertices for each edge $\{u,v\}$ in the original graph, make them adjacent to $y_2^{uv}$, and add them to the set $S$.
This will result in an offset for the size of weakly $r$\nobreakdash-reachable sets of vertices $y_{r+1}^{uv}$.

\begin{sloppypar}
On the other hand, it is not hard to see that \icwrcolleftlong\ is in~XP: We can simply try placing any of the vertices from $V(G)\setminus S$ into the next free position right of $L_S$.
As there are~$c$ free positions the overall algorithm runs in $\mathcal{O}(|V(G)\setminus S|^c\cdot |V(G)|^{\mathcal{O}(1)})\subseteq \mathcal{O}(|V(G)|^c\cdot |V(G)|^{\mathcal{O}(1)})$ time.\end{sloppypar}
\begin{proposition}
  \label{prop:icwrcolleftxp}
  \icwrcolleftlong\ parameterized by the reconstruction parameter~$c$ is in XP.
\end{proposition}
Our algorithm for \icwrcolleftlong\ is based on \cref{prop:icwrcolleftxp}, we will go into more detail in \cref{sec:impexp}.

\subsection{WCOL-Merge(\texorpdfstring{\boldmath$r$}))}
It is easy to see that \wcolmerge\ is NP-hard for $r\ge 2$ by giving a reduction from \wrcolproblong:
Given an instance~$(G,k)$ of \wrcolproblong, we create an equivalent instance $(H,S_1,S_2,L_{S_1})$ of \wcolmerge\ by setting $S_1=L_{S_1}=\emptyset$, $H=G$, and $S_2=V(G)$.
As deciding $\wcol_r(G)\le k$ is NP-hard for $r\ge 2$~\cite{DBLP:journals/siamdm/GroheKRSS18,DBLP:journals/corr/abs-2112-10562}, so is \wcolmerge.
On the positive side, we can show fixed-parameter tractability of \wcolmerge\ parameterized by~$k$ and~$|S_2|$.
\begin{theorem}
  \label{thm:mergefpt}
  \wcolmerge\ is solvable in time $\mathcal{O}(|S_2|!\cdot k^{|S_2|}\cdot |V(G)|^{\mathcal{O}(1)})$.
  In particular, \wcolmerge\ is fixed-parameter tractable when parameterized by $k+|S_2|$.
\end{theorem}
\noindent
Intuitively, the algorithm behind \cref{thm:mergefpt} tries to place each vertex~$v$ in $S_2$ one by one by trying all relevant positions.
The key insight is that only few positions are relevant (called breakpoints below).
Namely, those positions that correspond to vertices that are weakly $r$\nobreakdash-reachable from~$v$ when placed at the end of the ordering.
As only few vertices can be reachable from $v$ when placed in the correct position, we only need to try the first~$k$ corresponding positions.

To describe the algorithm we need definitions for two operations that we use throughout.
\begin{definition}
  \label{def:placeafter}
  Let~$G$ be a graph, $L_S=(s_1,\dots,s_n)$ a subordering, and $v\in V(G)\setminus S$.
  We denote by $\placeafter(L_S,s_i,v)$ the subordering of vertices $S\cup \{v\}$ that is obtained by placing~$v$ directly after $s_i$.
  To be precise, $\placeafter(L_S,s_i,v) := (s_1,\dots,s_i,v,s_{i+1},\dots ,s_n)$.
  Equivalently, $\placebefore(L_S,s_i,v) := (s_1,\dots,s_{i-1},v,s_{i},\dots ,s_n)$.
\end{definition}
This leads to the definitions of breakpoints, which are crucial for the proof of \cref{thm:mergefpt}.
\begin{definition}
  Let~$G$ be a graph, $L_S=(s_1,\dots ,s_n)$ a subordering, and $v\in V(G)\setminus S$.
  A vertex $s\in S$ is called \emph{breakpoint} of~$v$ if
  \[\wreach_r(G,\placebefore(L_S,s,v),v)\ne \wreach_r(G,\placeafter(L_S,s,v),v).\]
  Let $\mathrm{bp}(G,L_S,v)\subseteq S$ be the set of breakpoints of~$v$.
\end{definition}
We also notice another useful property of breakpoints.
\begin{lemma}
  \label{lemma:fptmerge}
  Let $v\in V(G)\setminus S$. We have $s\not\in \mathrm{bp}(G,L_S,v)$ if and only if for all $u\in V(G)$
  \[\wreach_r(G,\placebefore(L_S,s,v),u)=\wreach_r(G,\placeafter(L_S,s,v),u).\]
\end{lemma}
\begin{proof}
  Let $L^a=\placebefore(L_S,s,v)$ and $L^b=\placeafter(L_S,s,v)$
  If 
  \[\forall u\in V(G):\wreach_r(G,L^b,u)=\wreach_r(G,L^a,u),\] 
  then $s\not\in \mathrm{bp}(G,L_S,v)$ by setting $u=v$ and by the definition of breakpoints. 
  
  Thus, assume that there exists $u\in V(G)$ such that 
  \[\wreach_r(G,L^b,u)\ne\wreach_r(G,L^a,u).\]
  If $u=v$ we are done. Otherwise, there exists $x\in S\cup \{v\}$ such that either $x\in \wreach_r(G,L^b,u)\setminus\wreach_r(G,L^a,u)$ or $x\in \wreach_r(G,L^a,u)\setminus\wreach_r(G,L^b,u)$. In both cases there exists a path $P$ that witnesses $x\in \wreach_r(G,L^a,u)$ or $x\in \wreach_r(G,L^b,u)$. 
  If $V(P)\cap \{s,v\}\ne \{s,v\}$ then $L^a[V(P)]=L^b[V(P)]$, yielding a contradiction to $x\not\in \wreach_r(G,L^b,u)\cap\wreach_r(G,L^a,u)$.
  Hence, $v,s\in V(P)$.
  Furthermore, $v$ and $s$ must be the leftmost vertices on $P$ w.r.t.\ $L^a$ and $L^b$, as otherwise $x\in \wreach_r(G,L^b,u)\cap\wreach_r(G,L^a,u)$ again.
  When considering the subpath of $P$ between $v$ and $s$, we get $s\in \wreach_r(G,L^a,v)$, but we cannot have $s\in \wreach_r(G,L^b,v)$ as $v\prec_{L^b} s$.
  It follows that $s\in \mathrm{bp}(G,L_S,v)$ which completes the proof.
\end{proof}

\looseness=-1
If we add a vertex~$v$ to a subordering~$L_{S}$ to obtain a new subordering $L_{S\cup\{v\}}$, then the weakly reachable vertices $\wreach_r(G,L_{S\cup \{v\}}, v)$ consist of~$v$ and a subset of~$\mathrm{bp}(G,L_S,v)$.
We formalize this as follows.
\begin{lemma}
  \label{lemma:breakpoints}
  Let~$G$ be a graph, $L_S=(s_1,\dots ,s_n)$ be a subordering, and $v\in V(G)\setminus S$.
  Furthermore, let~$L_{S\cup \{v\}}$ be a subordering such that $L_{S\cup \{v\}}[S]=L_S$.
  Then
  \[\wreach_r(G,L_{S\cup \{v\}},v)\setminus\{v\}=\{s\in \mathrm{bp}(G,L_S,v)\mid s\preceq_{L_{S\cup \{v\}}}v\}.\]
\end{lemma}
\begin{proof}
  Let $X$ be the set $\{s\in \mathrm{bp}(G,L_S,v)\mid s\preceq_{L_{S\cup \{v\}}}v\}$, we prove both inclusions of the equation $\wreach_r(G,L_{S\cup \{v\}},v)\setminus\{v\}=X$.

  Assume that $s\in(\wreach_r(G,L_{S\cup \{v\}},v)\setminus\{v\})$.
  As~$s$ is weakly $r$\nobreakdash-reachable from~$v$, there is a path $P=(v,u_1,\dots,u_\ell,s)$ of length of at most~$r$ that does not go left of~$s$ w.r.t.\ to $L_{S\cup \{v\}}$.
  Consider the same path~$P$ in $\placeafter(L_S,s,v)$.
  Clearly,~$s$ is also weakly $r$\nobreakdash-reachable in this subordering because of the same path.
  Contrary to that,~$s$ cannot be weakly $r$\nobreakdash-reachable from~$v$ in $\placebefore(L_S,s,v)$ because~$v$ is left of~$s$ in that subordering.
  Hence, $s\in \mathrm{bp}(G,L_S,v)$ and $\wreach_r(G,L_{S\cup \{v\}},v)\setminus\{v\}\subseteq X$.

  Assume that $s\in X$. Then~$s$ must be weakly $r$\nobreakdash-reachable from~$v$ w.r.t.\ $\placeafter(L_S,s,v)$ through a path~$P$ of length at most~$r$.
  But~$s$ is also weakly $r$\nobreakdash-reachable from~$v$ w.r.t.~$L_{S\cup \{v\}}$ through the same path~$P$.
  Hence, $X\subseteq \wreach_r(G,L_{S\cup \{v\}},v)\setminus\{v\}$ also holds.
\end{proof}
\begin{algorithm}[!t]
  \SetKw{False}{\textsf{\textup{false}}}
  \SetKwFunction{FMain}{Recursive-merge}
  \SetKwProg{Fn}{}{:}{}
  \Fn{\FMain{$S_1$, $S_2$, $T$, $L_{S_1}$}}{
  \lIf{$|S_2|=0\land \forall v\in V(G):|\wreach_r(G,L_{S_1},v)|\le k$}{
    \Return{$L_{S_1}$}
  }
  \For{$v\in S_2$}{
  \For{$s\in \mathrm{bp}(G[S_1\cup T\cup \{v\}],L_{S_1},v)$}{
  $L_{S_1\cup \{v\}}\gets\placebefore(L_{S_1},s,v)$;\\
  \If{$|\wreach_r(G[S_1\cup T\cup \{v\}],L_{S_1\cup \{v\}},v)|\le k$}{\label{line:fptmergeif1}
  $\textup{\textsf{answer}} \gets {}$\FMain{$S_1\cup \{v\},S_2\setminus \{v\},T,L_{S_1\cup \{v\}}$};\label{line:fptmergerec1}\\
  \lIf{\textup{\textsf{answer}}${}\ne{}$\False}{\Return{\textup{\textsf{answer}}}}
  }
  }
  $s\gets$rightmost vertex of $S_1$ w.r.t.~$L_{S_1}$;\\
  $L_{S_1\cup \{v\}}\gets\placeafter(L_{S_1},s,v)$;\\
  \If{$|\wreach_r(G[S_1\cup T\cup \{v\}],L_{S_1\cup \{v\}},v)|\le k$}{\label{line:fptmergeif2}
  \textup{\textsf{answer}}${}\gets{}$\FMain{$S_1\cup \{v\},S_2\setminus \{v\},T,L_{S_1\cup \{v\}}$};\label{line:fptmergerec2}\\
  \lIf{\textup{\textsf{answer}}${}\ne{}$\False}{\Return{\textup{\textsf{answer}}}}
  }
  }
  \Return{\False}}
  \caption{Recursive FPT-algorithm for \wcolmerge}
  \label[algorithm]{alg:mergefpt}
\end{algorithm}

Using the above tools, we can now formally describe \cref{alg:mergefpt}, which obtains the stated runtime of \cref{thm:mergefpt} and is given as a recursive function \textsc{Recursive-merge}.
As alluded to before, the intuition is that for each vertex $v\in S_2$ we only have to consider placing it before its breakpoints w.r.t.~$L_{S_1}$.
As the breakpoints of a vertex will be in its weakly $r$\nobreakdash-reachable set, only the leftmost~$k$ breakpoints are relevant.
\begin{proof}[Proof of \cref{thm:mergefpt}]
  We give a recursive search tree algorithm for \wcolmerge\ in time $\mathcal{O}(|S_2|!\cdot k^{|S_2|}\cdot n_G^{\mathcal{O}(1)})$.
  This algorithm is given in \cref{alg:mergefpt} as a recursive function.
  Given an instance $(G,S_1,S_2,L_{S_1})$ of \wcolmerge, the given function can be invoked as \textsc{Recursive-merge}($S_1,S_2,V(G)\setminus(S_1\cup S_2),L_{S_1}$), and will either return an extendable subordering $L_{S_1\cup S_2}$ such that $L_{S_1\cup S_2}[S_1]=L_{S_1}$, or it will return \textsf{false} if no such subordering exists.
  The set $V(G)\setminus (S_1\cup S_2)$ is the set of free vertices~$T$.
  In a recursive call, the algorithm places one vertex $v\in S_2$ into~$L_{S_1}$ and thus decreases the cardinality of $S_2$ and increases the cardinality of $S_1$.
  Furthermore, in each recursive call the following new calls are created for each vertex $v\in S_2$.
  \begin{enumerate}
    \item[(1)] For each vertex $s\in \mathrm{bp}(G[S_1\cup T\cup\{v\}],L_{S_1},v)$ consider $L_{S_1\cup\{v\}}=\placeafter(L_{S_1},s,v)$.
          If $|\wreach_r(G[S_1\cup T\cup\{v\}],L_{S_1\cup\{v\}},v)|\le k$, then create a new recursive call $(S_1\cup \{v\},S_2\setminus \{v\},T,L_{S\cup\{v\}})$.
    \item[(2)] Let $L_{S_1\cup\{v\}}$ be the subordering obtained from~$L_{S_1}$ by placing~$v$ at the right end.
          If we have that $|\wreach_r(G[S_1\cup T\cup\{v\}],L_{S_1\cup\{v\}},v)|\le k$, then create a new recursive call $(S_1\cup \{v\},S_2\setminus \{v\},T,L_{S_1\cup\{v\}})$.
  \end{enumerate}
  Notice that we are only considering the induced subgraph $G[S_1\cup T\cup\{v\}]$, as we do not know where the remaining vertices of $S_2$ will be placed.
  We want to show that in this way only~$k$ new recursive calls are created for each vertex $v\in S_2$ and that the algorithm finds a solution (subordering $L_{S_1\cup S_2}$) if there exists one.

  \proofparagraph{Correctness.}
  Assume that we are in a recursive call $(S_1,S_2,T,L_{S_1})$ of the algorithm and that there exists an extendable subordering $L_{S_1\cup S_2}$ such that $L_{S_1\cup S_2}[S_1]=L_1$.
  We show that a (possibly different) subordering $L_{S_1\cup S_2}^\prime$ with these properties can be found in one of the generated new recursive calls.
  Consider the rightmost vertex~$v\in S_2$ w.r.t.\ $L_{S_1\cup S_2}$.
  There are two possibilities.
  \begin{itemize}
    \item There is a breakpoint of~$v$ (w.r.t.~$L_{S_1}$) right of (w.r.t.\ $L_{S_1\cup S_2}$)~$v$, that is $\mathrm{bp}(G[S_1\cup T\cup\{v\}],L_{S_1},v)\cap\{s\in L_{S_1}\mid v\prec_{L_{S_1\cup S_2}}s\}\ne \emptyset$.
          Then consider the leftmost breakpoint $s\in \mathrm{bp}(G[S_1\cup T\cup\{v\}],L_{S_1},v)$ that is right of~$v$ w.r.t.\ $L_{S_1\cup S_2}$.
          We shift~$v$ left of~$s$, creating a new subordering $L_{S_1\cup S_2}^\prime$ where the weakly $r$\nobreakdash-reachable sets of all vertices do not change because of \cref{lemma:fptmerge};
          that is, we replace $(v,u_1,\dots,u_\ell,s)$ by $(u_1,\dots,u_\ell, v,s)$, where $(v,u_1,\dots,u_\ell,s)$ is the consecutive part of $L_{S_1\cup S_2}$ between~$v$ and~$s$.
          Weakly $r$\nobreakdash-reachable sets cannot change because none of $u_1,\dots,u_\ell$ are in $\mathrm{bp}(L_{S_1},v)$, and~$v$ is the rightmost vertex of $S_2$ w.r.t.\ $L_{S_1\cup S_2}$.
          Clearly, $L_{S_1\cup S_2}^\prime$ can be found in a recursive call generated from (1) in Line~\ref{line:fptmergerec1} because we try all breakpoints for every vertex $v\in S_2$.
    \item There is no breakpoint of~$v$ (w.r.t.~$L_{S_1}$) right of (w.r.t.\ $L_{S_1\cup S_2}$)~$v$, that is, $\mathrm{bp}(G[S_1\cup T\cup\{v\}],L_{S_1},v)\cap\{s\in L_{S_1}\mid v\prec_{L_{S_1\cup S_2}}s\}= \emptyset$.
          Then we can shift~$v$ to the end of $L_{S_1\cup S_2}$, creating a new subordering $L_{S_1\cup S_2}^\prime$ where weakly $r$\nobreakdash-reachable sets do not change;
          that is, we replace $(v,u_1,\dots,u_\ell)$ by $(u_1,\dots,u_\ell,v)$, where $(v,u_1,\dots,u_\ell)$ is the consecutive part of $L_{S_1\cup S_2}$ right of~$v$.
          Again, weakly $r$\nobreakdash-reachable sets cannot change because none of $u_1,\dots,u_\ell$ are in $\mathrm{bp}(L_{S_1},v)$, and~$v$ is the rightmost vertex of~$S_2$ w.r.t.\ $L_{S_1\cup S_2}$.
          The subordering $L_{S_1\cup S_2}^\prime$ can be found in a recursion call generated from (2) in Line~\ref{line:fptmergerec2}.
  \end{itemize}
  Notice that none of the if-statements in Line~\ref{line:fptmergeif1} and Line~\ref{line:fptmergeif2} of \cref{alg:mergefpt} contradict the correctness because the cardinality of the weakly $r$\nobreakdash-reachable set of~$v$ can only increase in subsequent recursive calls.

  \proofparagraph{Runtime.} We show that in each recursion call at most $k\cdot |S_2|$ new recursion calls are created.
  As the size of $S_2$ decreases by one in each of these new calls, the stated runtime follows.
  But this is a direct consequence of \cref{lemma:breakpoints}, as if a vertex~$v$ is placed to the right of one of its breakpoints $s\in \mathrm{bp}(G[S_1\cup T\cup\{v\}],L_{S_1},v)$, then~$s$ will be in the weakly $r$\nobreakdash-reachable set of~$v$.
  As we do not recurse if the size of the weakly $r$\nobreakdash-reachable set of~$v$ exceeds~$k$ (if-statements in Line~\ref{line:fptmergeif1} and Line~\ref{line:fptmergeif2}), we generate at most~$k$ recursive calls for each vertex $v\in S_2$.
  The rest of the algorithm can be implemented in polynomial time, resulting in an overall runtime bounded by $\mathcal{O}(|S_2|!\cdot k^{|S_2|}\cdot |V(G)|^{\mathcal{O}(1)})$.
\end{proof}

\section{Implementation and experiments}\label{sec:impexp}
This section contains implementation details for the heuristics and both turbocharging algorithms.
Furthermore, we give the experimental setup and experimental results.
\subsection{Algorithm implementations and heuristic improvements}
Interesting details are omitted from the algorithmic results of \cref{sec:complexity}.
We want to give some implementation details for \icwrcolleft\ and \wcolmerge, and how the algorithms for these problems are combined with the heuristics.
Additionally, we give a third turbocharging approach called \textsc{IC-WCOL-RL($r$)} in this section.

In our implementations of heuristics and turbocharging algorithms we store and update the current subordering~$L_S$ in a simple array.
We also store and update the sets $\wreach_r(G,L_S,v)$ and $\wreach^{-1}_r(G,L_S,v)=\{w\in V(G):v\in\wreach_r(G,L_S,w)\}$ (see below for their usage).
Updating weakly $r$\nobreakdash-reachable sets during placements and removals of vertices~$v$ is done by computing the set $\wreach_r^{-1}(G,L_S,v)$ via a breadth-first search that respects the order~$L_S$, and updating the corresponding weakly $r$\nobreakdash-reachable sets.

\looseness=-1
Additionally, we slightly adapt the Degree-Heuristic:
We aim for vertex orderings $L$ with $\wcol_r(G,L)\le k$.
Consider a subordering $L_S$ that was created by the heuristic and needs to be extended.
To obtain weak coloring number~$k$ it would intuitively make sense to place a free vertex~$v$ with $\wcol_r(G,L_S)=k$ immediately to the right of that subordering s.t.\ its weakly $r$\nobreakdash-reachable set cannot increase anymore.
This is indeed always correct --- if there is a full right extension $L$ of $L_S$ with $\wcol_r(G,L)\le k$, then there is another one that starts by placing $v$ immediately to the right of $L_S$ (see \cref{prop:fullverticesnext} of \cref{appendix:theory} for a formal statement and a proof).
In our implementation of the Degree-Heuristic we apply this observation and place such a vertex $v$ immediately.
The Wreach-Heuristic does this implicitly.

We continue by explaining individual details for \icwrcolleft\ and \wcolmerge, and how they are applied to a non-extendable subordering~$L_S$ with free vertices~$T$.

\subparagraph*{IC-WCOL(\boldmath$r$).}
We have implemented the XP-algorithm for \icwrcolleftlong\ as outlined in \cref{prop:icwrcolleftxp}.
Given a subordering~$L_S$, we have to extend~$L_S$ to the right by~$c$ vertices, that means that we have~$c$ positions to fill.
We implement a search tree algorithm that fills these positions from left to right recursively.
That is, in a search tree node we try all possibilities of placing a free vertex into the leftmost free position~$i$ and recurse into search tree nodes that try placing the remaining free vertices into position $i+1$, and so on, until all~$c$ positions are filled.

If after a placement of a vertex we obtain a non-extendable subordering, we can cut off this branch of the search tree, as weakly $r$\nobreakdash-reachable sets of vertices can only increase in this branch.
We also store edges of~$G[T]$ separately as an array of hash sets.
This enables, in a search tree node, to update the sets $\wreach_r^{-1}(G,L_S,v)$ on placement/removal by a simple depth-$r$ breadth-first-search in~$G[T]$. This decreases the number of enumerated edges compared to the trivial approach.

Free vertices~$T$ are placed into a position~$i$ in a specific order inside a search tree node:
let~$L_S$ be the non-extendable subordering that triggered turbocharging and let~$v$ be the rightmost vertex of~$L_S$.
We try placing $u_1\in T$ before $u_2\in T$ into the free slot~$i$ if $\dist_G(u_1,v)<\dist_G(u_2,v)$.
Preliminary experiments suggested that this order is preferable to a random one.
We compute $\dist_G(u,v)$ for all $u,v$ with Johnson's Algorithm \cite{DBLP:books/daglib/0023376} for sparse graphs once in the beginning.

\subparagraph*{WCOL-Merge(\boldmath$r$).}
We apply \wcolmerge\ to a non-extendable subordering~$L_S$ in the following way.
Let $U=\{v\in V(G):|\wreach_r(G,L_S,v)|>k\}$ be the set of~\emph{overfull} vertices.
Let~$c$ be a positive integer and let $X$ be a random subset of $\bigcup_{v\in U}\wreach_r(G,L_S,v)$ of size $\min(c,|\bigcup_{v\in U}\wreach_r(G,L_S,v)|)$.
If the size of~$X$ is less than~$c$, we randomly add additional vertices from~$V(G)$ to~$X$, until the size of~$X$ is~$c$.
We try to fix~$L_S$ by defining an instance of \wcolmerge.
Namely, we set $S_1=S\setminus X$, $S_2=X$, and $L_{S_1}=L_S[S_1]$.
We then solve this instance using \cref{alg:mergefpt}.
By \cref{thm:mergefpt} we obtain a turbocharging algorithm that has fixed-parameter tractable running time when parameterized by the desired coloring number~$k$ and reconstruction parameter~$c$.
In our implementation we apply \wcolmerge\ multiple times with different randomly selected sets~$X$ as defined above, until we obtain an extendable subordering.
Preliminary experiments showed that choosing the whole set $\bigcup_{v\in U}\wreach_r(G,L_S,v)$ as~$X$ leads to timeouts often whereas random subsets still allowed us to fix~$L_S$.
If we do not find an extendable subordering after the 10th application of \wcolmerge, we report that turbocharging was not successful.

\looseness=-1
We now discuss the implementation of \cref{alg:mergefpt}.
Consider the vertex~$v$ in \cref{alg:mergefpt}.
We only have to iterate over the~$k$ leftmost breakpoints of~$v$ due to \cref{lemma:breakpoints}, which can be easily done by storing and maintaining $\wreach_r^{-1}(G[S_1\cup T\cup\{v\}],L_{S_1\cup T}, v)$.
Let~$s$ be a breakpoint of~$v$ and let $L_{S_1\cup\{v\}}=\placeafter(L_{S_1},s,v)$.
The leftmost $s^\prime\in \wreach_r^{-1}(G[S_1\cup T\cup\{v\}],L_{S_1\cup\{v\}},v)$ that is not~$v$ is the next possible breakpoint of~$v$.
Additionally, we do not need to recurse if the size of some set $\wreach_r(G[S_1\cup T],L_{S_1},v)$ exceeds~$k$ for some $v\in S_1\cup T$, as these sets can only increase in subsequent recursion calls.

We also know that subsets of some weakly $r$\nobreakdash-reachable sets of vertices~$T$ are already fixed.
Namely, for all vertices~$v$ in the $r$\nobreakdash-neighborhood of~$u$ in~$G[T\cup \{u\}]$ with $u\in S_2$, vertex~$u$ will always be in the weakly $r$\nobreakdash-reachable set of~$v$ if~$u$ is placed somewhere into the subordering~$L_{S_1}$.
We take this into account when calculating lower bounds for the sizes of weakly $r$\nobreakdash-reachable sets of vertices in~$T$ (and break the search if they exceed size~$k$).

\subparagraph*{IC-WCOL-RL(\boldmath$r$).}\looseness=-1
We also implemented a turbocharging algorithm that is not discussed above.
It is based on the \emph{Sreach-Heuristic}, which builds a vertex ordering of low weak coloring number from right to left (instead of from left to right as above)~\cite{DBLP:journals/jea/NadaraPRRS19}.
It starts with an empty subordering~$L_S$ and, during each step, the heuristic adds to the left front of $L_S$ a free vertex~$v$ that minimizes the number of so-called potentially strongly $r$-reachable vertices after being placed.
Herein, a vertex $u$ is \emph{potentially strongly $r$-reachable} w.r.t.\ $L_S$ from a vertex $v \in S$ if $u \in \wreach_r(G[S], L_S, v)$ or there is a path $P$ of length at most~$r$ from $v$ to $u$ in $G$ such that $V(P) \cap T = \{u\}$.
It can be shown (refer to a formal statement and a proof in \cref{lem:leftextfixed} in the appendix) that the set of potentially strongly $r$-reachable vertices of $v$ only grows when extending $L_S$ to the left and that, when $S = V(G)$, this set equals $\wreach_r(G, L_S, v)$.
Thus, we define a \emph{point of regret} for this heuristic as a point in the execution where there is a vertex~$v$ such that the size of its potentially strongly $r$-reachable set exceeds the desired weak coloring number~$k$.
Accordingly, we say that $L_S$ is \emph{non-extendable}, and otherwise it is \emph{extendable}.
We then solve the turbocharging problem in which we aim to replace the~$c$ leftmost vertices of~$L_S$ with arbitrary free vertices in order to make $L_S$ extendable again.
We do this using a search-tree algorithm analogous to IC-WCOL($r$).
We also use the above turbocharging approach with a heuristic that chooses the next vertex among the free vertices based on the smallest degree. We call this heuristic Degree-Heuristic\footnote{The name is the same as in the left-to-right setting; there will be no confusion between the two because the direction will be clear from the context.}.
\looseness=-1
We tried further heuristic optimizations, mainly for the left-to-right approaches, based on lower bounds (for early search termination), guided branching (towards faster decomposition into trivial instances), and ordered adjacency lists (to speed up computation of weakly reachable sets) but they did not improve the resulting coloring numbers.

\subsection{Experiments setup}
\looseness=-1
\subparagraph*{Computation environment.}
All experiments were performed on a cluster of 20 nodes.
Each node is equipped with two Intel Xeon E5-2640 v4, 2.40GHz 10-core processors and 160 GB RAM.
The optimization for each instance and an algorithm was pinned to a specific core of a cluster node (simultaneous multithreading was disabled).
\looseness=-1
All implementations of heuristics and turbocharging algorithms were done in C\texttt{++}17, and made use of the Boost library\footnote{\url{https://www.boost.org/}}, version 1.77.0.
The code was compiled on Linux with \texttt{g++} version 7.5.0 and with the flags \texttt{-std=c++17 -O2}.
The optimization process (see \emph{application of turbocharging} below) that calls the heuristics combined with the turbocharging algorithms (implemented in C\texttt{++}) was implemented in Python3 and executed with Python 3.7.13.
A memory limit of 16 GB was set (a process only starts if the required memory is free).
The source code is available online \cite{doblalex_2021}.

\subparagraph*{Instances.}
Each instance in our data set is a tuple consisting of a graph~$G$ and a radius~$r \in \mathbb{N}$.
The radii~$r$ are between $2$ and $5$, as also used by Nadara et al.~\cite{DBLP:journals/jea/NadaraPRRS19}.
The graphs $G$ form a subset of the graphs used by Nadara et al.
This enables us to use weak coloring numbers of orderings computed by Nadara et al.\ as a baseline.
Furthermore, we can compare for a heuristic, the improvement achieved by the local search of Nadara et al.\ to the improvement achieved by our turbocharging algorithms.

\looseness=-1
The graphs consist of real-world data, the PACE 2016 Feedback Vertex Set problems, random planar graphs, and random graphs with bounded expansion.
Nadara et al.\ classified the graphs into four classes based on the number of edges --- small (up to 1k edges), medium (up to 10k edges), big (up to 48k edges), and huge.
\cref{table:test-graphs} contains some basic statistics of the graphs. 
For a detailed explanation and references for all input graphs we refer to Nadara et al.~\cite{DBLP:journals/jea/NadaraPRRS19}; the instances are available online\footnote{\url{https://kernelization-experiments.mimuw.edu.pl/}}.
We considered all instances except those where one of the three heuristics of Nadara et al.\ that we also consider did not yield a result (the Degree-Heuristic, Wreach-Heuristic, and Sreach-Heuristic).
That is, they timed out after 300 seconds or ran into a memory limit (16 GB).
In total, our dataset contains 334 instances.

\subparagraph*{Application of turbocharging.}\label{par:optbytc}
\begin{algorithm}[t]
  \KwIn{A graph $G$, an integer $r$, a heuristic $\mathcal{H}$, and a turbocharged heuristic $TC$-$\mathcal{H}$}
  \KwOut{An ordering~$L$ of vertices~$V$}
  $L\gets$ ordering of vertices~$V$ computed by the heuristic $\mathcal{H}$;\\
  $k\gets \wcol_r(G,L)$;\\
  Start a timer; after $t$ seconds abort the program, and \textbf{return} the current value of~$L$\\
  \While{\textsf{\textup{true}}}{
    $c\gets 1$;\\
    \While{\textsf{\textup{true}}}{
      Try to compute an ordering of vertices~$V$ with weak $r$\nobreakdash-coloring number $k-1$ using $TC$-$\mathcal{H}$ with reconstruction parameter~$c$;\\
      If successful, assign this ordering to~$L$, set $k\gets\wcol_r(G,L)$, and \textbf{break};\\
      Otherwise, set $c\gets c+1$;\\
    }
  }
  \caption{Algorithm for iteratively decreasing the weak coloring number by turbocharging.}
  \label[algorithm]{alg:expframework}
\end{algorithm}
Our application of turbocharging with a heuristic $\mathcal{H}$ works as follows.
Given a graph $G$ and a radius $r$, we start with a run of $\mathcal{H}$ without turbocharging to produce a vertex ordering for $G$ with a baseline weak $r$-coloring number~$k$.
We then start a timer that runs for~$t$ seconds, aborting the rest of the algorithm when it terminates.
We then iteratively decrease $k$ and apply~$\mathcal{H}$ together with the turbocharging approach to try and find a vertex ordering for $G$ with weak $r$-coloring number at most~$k$.
In each such try, we start with the reconstruction parameter $c = 1$.
If no ordering with the desired weak $r$-coloring number was produced by the turbocharged heuristic, we increase $c$ by one and try again.
If an ordering with weak $r$-coloring number $k' \leq k$ was produced, we set $k = k' - 1$ and repeat the process.
The precise algorithm is given in \cref{alg:expframework}.

For our experiments we applied all compatible combinations of heuristics and the turbocharged versions to each instance.
Furthermore, we computed orderings for radii ranging from~$2$ to~$5$, motivated by Nadara et al.\ who used the same values.
We ran all experiments twice, once with timeout $t=300s$ and once with $t=3600s$.
Results with $t=300s$ are directly compared with the results of Nadara et al.\ who used 300 seconds as timeout.
Results with $t=3600s$ give us the ability to investigate the potential of turbocharging over longer periods of~time.

\subsection{Results}

\newcommand{\nad}{\textsf{bestNadara}}
\looseness=-1

\medskip
We now show to which extent our turbocharging approaches improve the results of heuristics, compare the achieved weak coloring numbers to the ones of Nadara et al., and provide observations about the performance of turbocharging.

\subparagraph*{Impact of turbocharging.}
Turbocharging has the advantage that investing gradually more time will yield gradually better results --- setting the reconstruction parameter to $c=|V(G)|$, we (in theory) even can provably obtain the optimum.
\cref{fig:absimpovertime} shows the cumulative sum of absolute improvements over time when comparing each turbocharged heuristic with the underlying plain heuristic.
That is, for a specific time~$t$, the cumulative sum of absolute improvements for a turbocharging algorithm $\mathcal{A}$ and a heuristic $\mathcal{H}$ is $\sum_{I\in\text{instances}}(k_{I,\mathcal{H}}-k_{I,\mathcal{A},\mathcal{H}, t})$, where $k_{I,\mathcal{H}}$ is the coloring number achieved by the heuristic and $k_{I,\mathcal{A},\mathcal{H}, t}$ is the coloring number achieved by the turbocharged heuristic after time $t$.
Note that the $y$-axes do not have the same ranges as the underlying heuristics have different performance levels.
\begin{figure}[!t]
  \centering
  \includegraphics[width=\textwidth]{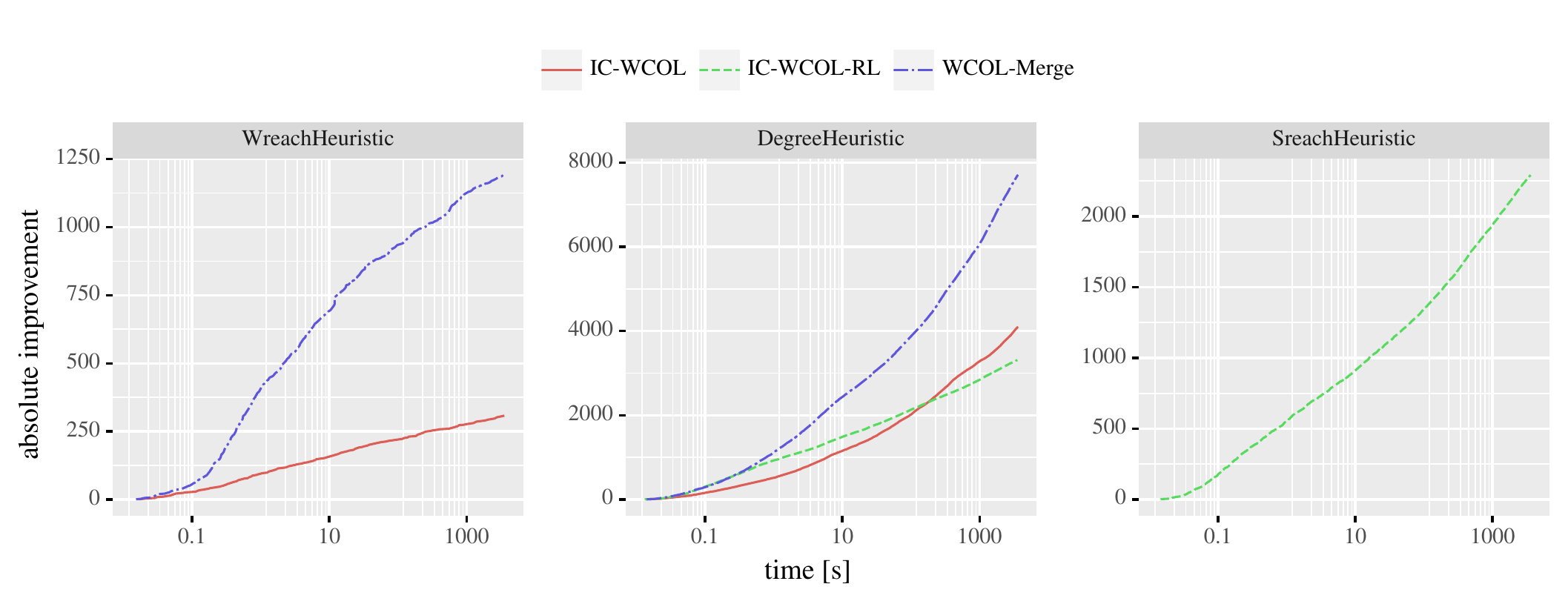}
  \caption{Line plot of the cumulative absolute improvements over time achieved by the turbocharging approaches broken down by the underlying heuristics. The $x$-axis (time) is scaled logarithmically.}
  \label{fig:absimpovertime}
\end{figure}
All plots exhibit logarithmic or similar to logarithmic growth, which means that the gained absolute improvement is approximately logarithmic in the invested time in the most cases.
\TCMERGE\ clearly yields faster and larger improvements than \TCLASTC, and it also supersedes \TCLASTCRL\ for the Degree-Heuristic.
One reason may be that \TCMERGE\ is fixed-parameter tractable and the associated parameters are~small.

\begin{sloppypar}
  We provide further evaluations of the executions of turbocharging algorithms in \cref{appendix:furtheranal}.
  Some observations therein are that the number of applications of turbocharging per run of a heuristic ranges in the order of at most hundreds for \TCLASTC\ and \TCMERGE\ and on average in the single digits; for \TCLASTCRL\ these numbers are one to two orders of magnitude larger.
  Successful applications of a turbocharged heuristic (those where the weak coloring number could be improved) mostly only use very little time and search tree nodes, and have small reconstruction parameters --- mostly $c=1$.
  That is, it is mostly the case that a heuristic wants to place a vertex that is ``suboptimal'', while placing nearly any other vertex will achieve lower weak coloring number.
  The fraction of time spent on turbocharging is also low, which means that most of the time when we can improve the weak coloring number achieved by a heuristic, this can be done easily and in little time.
\end{sloppypar}

\subparagraph{By dataset group and radius.}
Next, we fix a time threshold of 300s (same as Nadara et al.) that we might reasonably invest in practice for computing weak coloring numbers.
We present the improvements in weak coloring numbers gained by turbocharging over the plain heuristics broken down according to the instance group (small, medium, big, huge) and the radius~$r$.
We again provide absolute improvements when comparing with the underlying heuristic, and we also consider the average \emph{relative improvement} of the weak $r$\nobreakdash-coloring number when comparing the turbocharged heuristic to the plain heuristic; that is, the relative improvement is $1-k_{I,\mathcal{A},\mathcal{H},t}/k_{I,\mathcal{H}}$ for $t=300s$.
For each turbocharging algorithm we show results for both turbocharged heuristics.
For \TCLASTC\ and \TCMERGE\ these are the Wreach- and Degree-Heuristic, and for \TCLASTCRL\ these are the Sreach- and Degree-Heuristic.
The results are given in three tables corresponding to the three turbocharging approaches.

\looseness=-1
We furthermore compare the achieved coloring numbers of each approach to the best coloring numbers computed by Nadara et al.:
For each instance $I$, let $\nad(I)$ be the smallest weak $r$\nobreakdash-coloring number of an ordering of vertices of instance $I$ achieved by an approach of Nadara et al.
Note that they implemented seven different heuristics and for each computed ordering they applied a local search to iteratively reduce the weak $r$\nobreakdash-coloring number. %
To evaluate one of our approaches, we take the weak $r$-coloring number $k_{I,\mathcal{A},\mathcal{H},t}$ for instance $I$ obtained by our approach for $t=300s$ and compute the average $1 - k_{I,\mathcal{A},\mathcal{H},t}/\nad(I)$ (in percent) taken over all instances $I$ in the corresponding data set.
We call this value \emph{quality ratio}.
Note that positive values mean that the approach achieves lower weak coloring numbers on average when compared to the best weak coloring numbers achieved by Nadara et al.

\begin{table}[!t]
  \centering
  \caption{\TCLASTC: White columns give relative improvements, light gray columns give quality ratios, dark gray columns give average/maximum absolute improvements. For improvements, we compare to the underlying heuristic without turbocharging and without local search. Time limit: 300\,s.}
  \label{table:tclastcexp}
  \begin{tabularx}{.95\linewidth}{l c >{\centering\arraybackslash\columncolor{lightgray!15}}X >{\centering\arraybackslash\columncolor{darkgray!15}}X >{\centering\arraybackslash}X >{\centering\arraybackslash\columncolor{lightgray!15}}X >{\centering\arraybackslash\columncolor{darkgray!15}}X >{\centering\arraybackslash}X}
    \toprule
    tests                 & $r$ & \multicolumn{3}{c}{Wreach \TCLASTC} & \multicolumn{3}{c}{Degree \TCLASTC}                                                                    \\\midrule
    \multirow{4}*{small}  & 2   & -6.2\%                              & 0.8/5                               & \multirow{4}*{7.2\%} & -7.3\%  & 3.2/14  & \multirow{4}*{18.4\%} \\
                          & 3   & -7.3\%                              & 1.0/6                               &                      & -9.3\%  & 4.3/20  &                       \\
                          & 4   & -11.3\%                             & 1.7/7                               &                      & -11.2\% & 4.2/15  &                       \\
                          & 5   & -15.8\%                             & 1.9/8                               &                      & -16.8\% & 3.9/12  &                       \\
    \midrule
    \multirow{4}*{medium} & 2   & -5.6\%                              & 0.6/2                               & \multirow{4}*{3.7\%} & -7.1\%  & 5.0/14  & \multirow{4}*{17.4\%} \\
                          & 3   & -9.2\%                              & 1.0/11                              &                      & -9.7\%  & 9.0/35  &                       \\
                          & 4   & -11.6\%                             & 0.3/2                               &                      & -15.8\% & 9.9/31  &                       \\
                          & 5   & -16.8\%                             & 1.3/15                              &                      & -20.6\% & 9.9/41  &                       \\
    \midrule
    \multirow{4}*{big}    & 2   & -9.6\%                              & 0.1/1                               & \multirow{4}*{0.7\%} & -21.6\% & 7.2/31  & \multirow{4}*{12.0\%} \\
                          & 3   & -9.5\%                              & 0.4/3                               &                      & -20.8\% & 15.3/47 &                       \\
                          & 4   & -12.7\%                             & 0.3/2                               &                      & -32.5\% & 14.5/42 &                       \\
                          & 5   & -38.3\%                             & 0.2/1                               &                      & -30.4\% & 13.7/38 &                       \\
    \midrule
    \multirow{4}*{huge}   & 2   & -2.0\%                              & 0.1/1                               & \multirow{4}*{0.2\%} & -39.1\% & 2.8/16  & \multirow{4}*{0.9\%}  \\
                          & 3   & -21.4\%                             & 0.1/1                               &                      & -35.0\% & 1.9/6   &                       \\
                          & 4   & -6.9\%                              & 0.0/0                               &                      & -29.9\% & 1.8/5   &                       \\
                          & 5   & -8.9\%                              & 0.3/1                               &                      & -16.5\% & 1.7/4   &                       \\
    \bottomrule
  \end{tabularx}
\end{table}

\looseness=-1
In \cref{table:tclastcexp} we present the performance of the \TCLASTC\ approach.
It is evident that the relative and absolute improvements achieved for the Degree-Heuristic is significantly higher than for the Wreach-Heuristic, although this is partly due to the fact that the Degree-Heuristic achieves worse results than the Wreach-Heuristic before turbocharging.
Relative and absolute improvements decrease for larger instances.

\begin{table}[!t]
  \centering
  \caption{\TCMERGE: White columns give relative improvements, light gray columns give quality ratios, dark gray columns give average/maximum absolute improvements. For improvements, we compare to the underlying heuristic without turbocharging and without local search. Time limit: 300\,s.}
  \label{table:tcmergeexp}
  \begin{tabularx}{.95\linewidth}{l c >{\centering\arraybackslash\columncolor{lightgray!15}}X >{\centering\arraybackslash\columncolor{darkgray!15}}X >{\centering\arraybackslash}X >{\centering\arraybackslash\columncolor{lightgray!15}}X >{\centering\arraybackslash\columncolor{darkgray!15}}X >{\centering\arraybackslash}X}
    \toprule
    tests                 & $r$ & \multicolumn{3}{c}{Wreach \TCMERGE} & \multicolumn{3}{c}{Degree \TCMERGE}                                                                     \\\midrule
    \multirow{4}*{small}  & 2   & -1.0\%                              & 1.4/5                               & \multirow{4}*{19.1\%} & 0.3\%   & 4.1/11  & \multirow{4}*{33.7\%} \\
                          & 3   & 3.3\%                               & 2.8/8                               &                       & 2.9\%   & 6.3/25  &                       \\
                          & 4   & 0.7\%                               & 4.2/11                              &                       & 2.3\%   & 7.2/23  &                       \\
                          & 5   & -1.2\%                              & 5.4/14                              &                       & -0.2\%  & 8.0/19  &                       \\
    \midrule
    \multirow{4}*{medium} & 2   & 2.4\%                               & 1.6/7                               & \multirow{4}*{15.0\%} & 1.1\%   & 6.4/16  & \multirow{4}*{32.7\%} \\
                          & 3   & 0.4\%                               & 2.7/15                              &                       & 0.6\%   & 11.6/46 &                       \\
                          & 4   & -0.2\%                              & 2.9/13                              &                       & -1.9\%  & 13.3/36 &                       \\
                          & 5   & -5.3\%                              & 4.2/16                              &                       & -5.5\%  & 15.5/50 &                       \\
    \midrule
    \multirow{4}*{big}    & 2   & -0.3\%                              & 1.7/7                               & \multirow{4}*{9.4\%}  & -7.5\%  & 11.0/32 & \multirow{4}*{24.2\%} \\
                          & 3   & -1.8\%                              & 2.1/9                               &                       & -8.1\%  & 23.4/65 &                       \\
                          & 4   & -4.2\%                              & 2.8/11                              &                       & -19.3\% & 26.6/63 &                       \\
                          & 5   & -27.2\%                             & 4.9/26                              &                       & -20.5\% & 26.6/67 &                       \\
    \midrule
    \multirow{4}*{huge}   & 2   & -0.6\%                              & 1.7/5                               & \multirow{4}*{1.0\%}  & -14.9\% & 28.4/84 & \multirow{4}*{12.2\%} \\
                          & 3   & -20.5\%                             & 1.8/5                               &                       & -21.9\% & 27.4/49 &                       \\
                          & 4   & -6.7\%                              & 0.5/2                               &                       & -27.6\% & 11.2/20 &                       \\
                          & 5   & -8.7\%                              & 1.0/2                               &                       & -15.5\% & 6.0/14  &                       \\
    \bottomrule
  \end{tabularx}
\end{table}%

\looseness=-1
\cref{table:tcmergeexp} contains the results for \TCMERGE.
Here, turbocharging achieves positive quality ratios for some instance classes and radii.
The relative and absolute improvements are much larger than for \TCLASTC, especially for the huge instances and the Degree-Heuristic.
It is also interesting that while the Degree-Heuristic generally computes orderings of higher weak coloring number than the Wreach-Heuristic, the turbocharged version of the Degree-Heuristic computes orderings of similar or even lower weak coloring numbers than the turbocharged version of the Wreach-Heuristic for the small and medium instances.
We do not see an obvious reason for that, but it could again indicate the power of the fixed-parameter algorithm.

\begin{table}[!t]
  \centering
  \caption{\TCLASTCRL: White columns give relative improvements, light gray columns give quality ratios, dark gray columns give average/maximum absolute improvements. For improvements, we compare to the underlying heuristic without turbocharging and without local search. Time limit: 300\,s.}
  \label{table:exptclastcrl}
  \begin{tabularx}{.95\linewidth}{l c >{\centering\arraybackslash\columncolor{lightgray!15}}X >{\centering\arraybackslash\columncolor{darkgray!15}}X >{\centering\arraybackslash}X >{\centering\arraybackslash\columncolor{lightgray!15}}X >{\centering\arraybackslash\columncolor{darkgray!15}}X >{\centering\arraybackslash}X}
    \toprule
    tests                 & $r$ & \multicolumn{3}{c}{Sreach \TCLASTCRL} & \multicolumn{3}{c}{Degree \TCLASTCRL}                                                                     \\\midrule
    \multirow{4}*{small}  & 2   & -2.1\%                                & 3.2/28                                & \multirow{4}*{15.9\%} & -7.5\%  & 3.1/11  & \multirow{4}*{22.0\%} \\
                          & 3   & 1.6\%                                 & 2.7/10                                &                       & -7.4\%  & 4.4/19  &                       \\
                          & 4   & 3.6\%                                 & 2.5/7                                 &                       & -7.9\%  & 5.2/21  &                       \\
                          & 5   & 3.3\%                                 & 3.2/8                                 &                       & -8.9\%  & 6.3/26  &                       \\
    \midrule
    \multirow{4}*{medium} & 2   & -8.8\%                                & 3.1/12                                & \multirow{4}*{10.7\%} & -14.8\% & 3.0/8   & \multirow{4}*{15.5\%} \\
                          & 3   & -3.8\%                                & 3.8/16                                &                       & -12.6\% & 5.7/25  &                       \\
                          & 4   & -0.6\%                                & 3.8/14                                &                       & -16.4\% & 5.8/17  &                       \\
                          & 5   & 1.4\%                                 & 4.3/11                                &                       & -18.5\% & 7.6/22  &                       \\
    \midrule
    \multirow{4}*{big}    & 2   & -14.8\%                               & 3.0/10                                & \multirow{4}*{6.4\%}  & -29.8\% & 4.2/12  & \multirow{4}*{10.1\%} \\
                          & 3   & -13.4\%                               & 7.1/26                                &                       & -25.8\% & 8.8/23  &                       \\
                          & 4   & -7.7\%                                & 7.4/19                                &                       & -28.2\% & 13.1/60 &                       \\
                          & 5   & -3.1\%                                & 6.5/19                                &                       & -28.2\% & 12.2/31 &                       \\
    \midrule
    \multirow{4}*{huge}   & 2   & -22.8\%                               & 14.4/47                               & \multirow{4}*{5.4\%}  & -26.3\% & 22.1/76 & \multirow{4}*{7.1\%}  \\
                          & 3   & -16.2\%                               & 11.5/21                               &                       & -29.6\% & 12.4/20 &                       \\
                          & 4   & -17.2\%                               & 5.5/13                                &                       & -27.4\% & 8.2/16  &                       \\
                          & 5   & 2.0\%                                 & 4.3/12                                &                       & -14.4\% & 4.0/10  &                       \\
    \bottomrule
  \end{tabularx}
\end{table}

\cref{table:exptclastcrl} contains the results for \TCLASTCRL.
Although the relative and absolute improvements of turbocharging the Degree-Heuristic are slightly higher, the quality ratios for the turbocharged version of the Sreach-Heuristic are significantly better.
This could imply that \TCLASTCRL\ struggles to turbocharge slightly worse heuristics such as the Degree-Heuristic.
Furthermore, we see that for the Sreach-Heuristic the quality ratios are better for larger radii.
The reason for this could be that the Sreach-Heuristic performs well for larger radii even before turbocharging.
We also notice that for the medium graph class the quality ratios get worse.
The reason is that the implementation of \TCLASTCRL\ is slightly more computationally expensive than for the other approaches.

\begin{figure}[!t]
  \includegraphics[width=\textwidth]{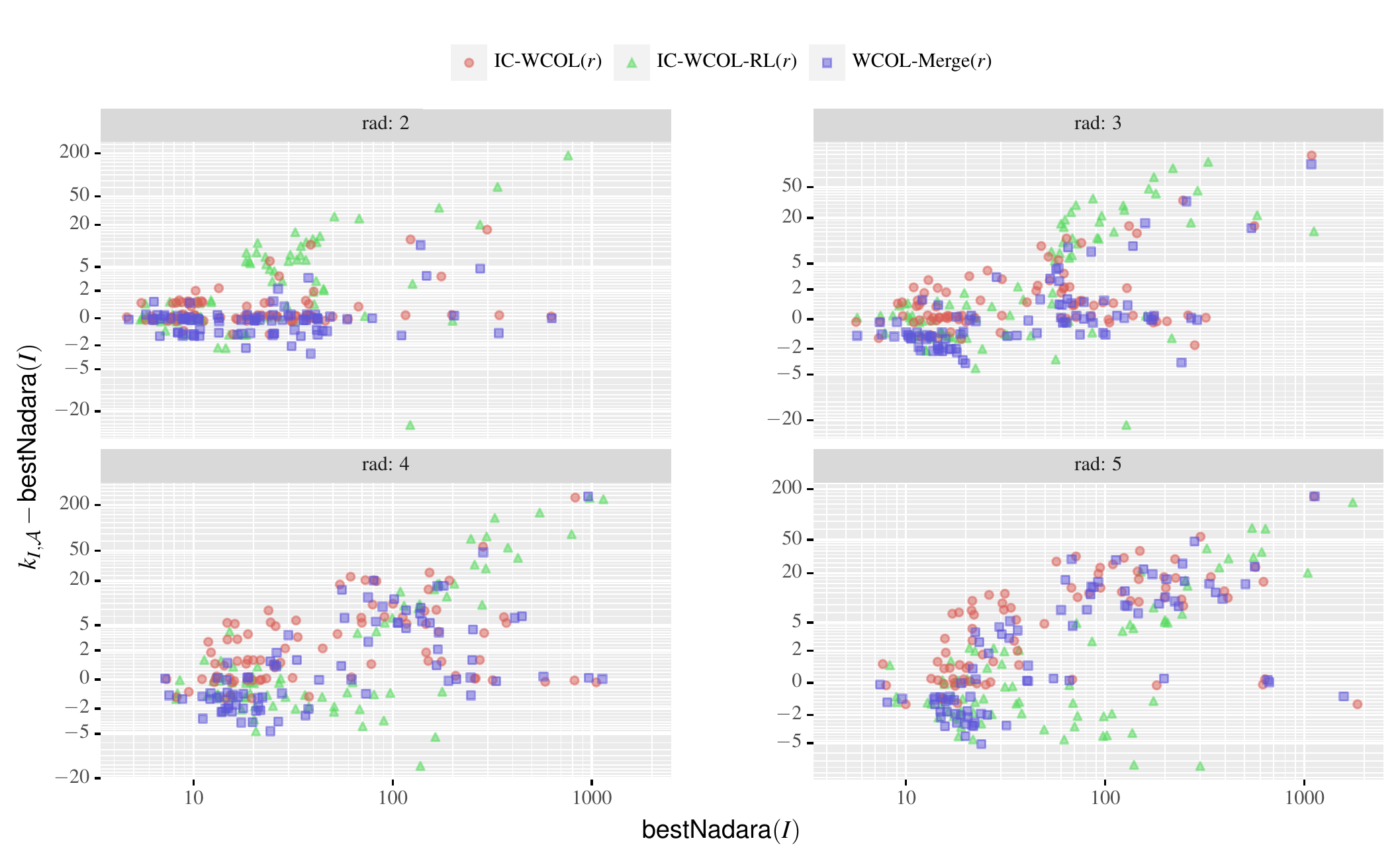}
  \caption[CaptionScatter]{Scatter plot of the results for turbocharging algorithms broken down by radius. Each data point corresponds to an instance and a turbocharging algorithm. The $x$-value is the best weak coloring number achieved by Nadara et al.\ for this instance, the $y$-value is the difference of coloring numbers between the best weak coloring number achieved by Nadara et al.\ and the weak coloring number achieved by the turbocharging algorithm (minimum over both turbocharged heuristics).
    The $x$-axis is scaled logarithmically and the $y$-axis is scaled pseudo-logarithmically\footnotemark. Time limit: 300\,s.}
  \label{fig:results}
\end{figure}

\subparagraph*{Achieved coloring numbers in comparison to Nadara et al.}
\cref{fig:results} illustrates the distribution of results for all turbocharging algorithms in a scatter plot.
Data points below $y$-value zero mean that the turbocharging algorithm improves a bound on the weak $r$\nobreakdash-coloring number of an instance.
Among our approaches, we see that while \TCMERGE\ performs well for small radii, \TCLASTCRL\ performs well for larger radii.
Together with the analysis from above we can conclude that Nadara et al.'s approaches yield lower weak coloring numbers on average, but there is fraction of roughly one half instances where our approaches supersede Nadara et al.'s, in particular if the computed weak coloring numbers are small.
Overall, we improved bounds for 172 of the 334 considered instances after $t=300s$.
The resulting new bounds are lower by 5\% on average over all instances.
Follow-up investigations also showed that the relative improvement of turbocharging negatively correlates with the average degree of the graph, suggesting that our turbocharging algorithms work better for particularly sparse graphs. We refer to \cref{fig:avgdegrelimp,fig:avgdegimpnad} in \cref{appendix:furtheranal} that illustrate this correlation for the average improvement over underlying heuristics and the relative improvements over the algorithms of Nadara et al.
We also tested if the relative improvement depends on the number of vertices, number of edges, or the average clustering coefficient \cite{wattsCollectiveDynamicsSmallworld1998}.
We were not able to observe a significant correlation for any of those parameters and the relative improvements.

\footnotetext{see \url{https://scales.r-lib.org/reference/pseudo_log_trans.html} for an explanation.}%
\begin{figure}[!t]
  \centering
  \includegraphics[width=.5\textwidth]{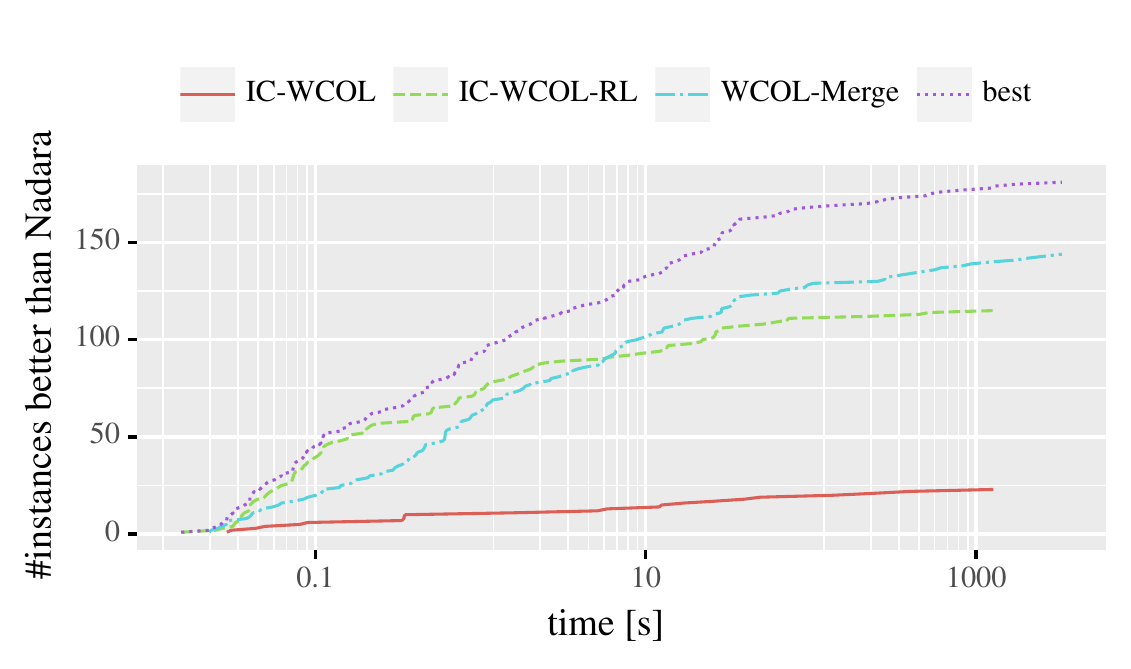}
  \caption{Line plot of the number of instances where a turbocharging approach achieves better coloring numbers than Nadara et al.\ over time. The time is scaled logarithmically.}
  \label{fig:absbetterthannad}
\end{figure}

As mentioned, our approach has the advantage that investing more time yields gradually better results.
\cref{fig:absbetterthannad} shows the number of instances for which a turbocharging approach improve weak coloring numbers compared to Nadara et al.\ after a specific time.
That is, for a time $t$, the $y$-value of a line corresponds to the number of instances~$I$ such that $k_{I,\mathcal{A},\mathcal{H}, t}<\nad(I)$.
The values for the line \emph{best} are determined by taking the number of instances $I$ where any of the turbocharging approaches is better than $\nad(I)$.
Interestingly, \TCLASTCRL\ starts off with more improved instances, however, after 3600 seconds, \TCMERGE\ achieved more instances with smaller coloring numbers than Nadara et al.
After 3600 seconds, \TCLASTC\ was able to improve 24 instances compared to Nadara et al., \TCMERGE\ 144 instances, and \TCLASTCRL\ 115 instances.
Overall, we could improve upper bounds for 181 of the 334 instances with $t=3600s$.
These are nine more than for $t=300s$.%

\section{Conclusion}
\looseness=-1
On the theoretical side, we determined obstructions (running-time lower bounds) and promising avenues (a fixed-parameter algorithm) for applying the turbocharging framework to computing vertex orderings of small weak coloring numbers.
On the experimental side, on a diverse set of instances each of the turbocharging approaches we use yields large improvements over the plain heuristics.
This is most pronounced for the fixed-parameter turbocharging \TCMERGE.
Then we compared turbocharging to the best results gained by the seven heuristics that Nadara et al.~\cite{DBLP:journals/jea/NadaraPRRS19} employed together with local search procedures.
Turbocharging so far yields on average larger weak coloring numbers than the best of the previous approaches.
However, for 173 of the in total 334 instances, turbocharging outperforms all of the previous approaches combined.
It works particularly well for small computed coloring numbers and sparse input instances.

\printbibliography

\clearpage

\appendix

\section{Appendix}
\subsection{Omitted theoretical results}\label{appendix:theory}
\begin{lemma}
  \label{lemma:nopaththrougS}
  Consider an extendable subordering~$L_S$ with free vertices~$T$.
  Assume that there is a full right extension $L^\prime$ of~$L_S$ with $\wreach_r(G,L^\prime)\le k$.
  For each vertex $u\in T$, let $X_{L^\prime}(u)=\{v\in T\mid v\prec_{L^\prime} u\}$ be
  the vertices that are placed between~$S$ and~$u$ w.r.t.\ $L^\prime$.
  If $\wreach_r(G,L^\prime,u)\cap X_{L^\prime}(u)=\emptyset$, then each path of length at most~$r$ between~$u$ and a vertex in $X_{L^\prime}(u)$ goes through~$S$.
\end{lemma}
\begin{proof}
  Assume to the contrary that there is a path $P$ of length at most~$r$ between~$u$ and a vertex in $X_{L^\prime}(u)$ that does not go through~$S$.
  Consider the leftmost (w.r.t.~$L^\prime$) vertex~$x$ on this path.
  Since there is a path from~$u$ to $x$ of length at most~$r$ that does not go left of $x$, $x$ has to be in $\wreach_r(G,L^\prime,u)\cap X_{L^\prime}(u)$, yielding a contraction.
\end{proof}
\begin{proposition}\label{prop:fullverticesnext}
  \sloppy
  Consider an extendable subordering~$L_S$ with free vertices~$T$.
  Assume that there is a full right extension $L^\prime$ of~$L_S$ with $\wreach_r(G,L^\prime)\le k$.
  If for $u\in T$, $|\wreach_r(G,L_S,u)|=k$, then there is another full right extension~$\overline{L}$ of~$L$ where~$u$ is the leftmost vertex of $T$ w.r.t.\ $\overline{L}$ and $\wreach_r(G,\overline{L})\le k$.
\end{proposition}
\begin{proof}
  \sloppy
  We construct $\overline{L}$ from $L^\prime$.
  Assume that $L^\prime=(s_1,\dots, s_n,t_1,\dots, t_m)$ with $S=\{s_1,\dots ,s_n\}$ and $T=\{t_1,\dots, t_m\}$;
  furthermore, assume $u=t_i$ with $1\le i\le m$.
  We set $\overline{L}=(s_1,\dots, s_n,u,t_1,\dots, t_{i-1},t_{i+1},\dots, t_m)$.
  The ordering $\overline{L}$ has the required properties because of \cref{lemma:nopaththrougS}:
  Assuming that there exists a vertex~$v$ such that $ | \wreach_r(G,\overline{L},v) |>| \wreach_r(G,L,v) |$, we will reach a contradiction;
  by shifting~$u$ left, we only increase weakly $r$\nobreakdash-reachable sets of vertices $t_j,j\ne i$, by adding~$u$ to their weakly $r$\nobreakdash-reachable sets.
  If~$u$ is weakly $r$\nobreakdash-reachable from a vertex $t_j$ with respect to $\overline{L}$, and it was not weakly $r$\nobreakdash-reachable with respect to $L^\prime$, then this new path goes through a vertex in $X_{L^\prime}(u)$ as defined in \cref{lemma:nopaththrougS} as it cannot go left of~$u$.
  Whenever $|\wreach_r(G,L_S,u)|=k$, then $X_{L^\prime}(u)\cap \wreach_r(G,L^\prime,u)=\emptyset$.
  Thus, every path of length at most~$r$ from~$u$ to a vertex in $X_{L^\prime}(u)$ must go through~$S$.
  But since all vertices in~$S$ are placed left of~$u$ with respect to $\overline{L}$,~$u$ cannot be reached from $t_j$, leading to a contradiction.
\end{proof}

\newcommand{\wreachright}{\mathrm{PotSreach}}
\begin{lemma}
  \label{lem:leftextfixed}
  Let $\wreachright_r(G,L_S,v)$ be the set of potentially strongly $r$-reachable vertices of $v\in S$ w.r.t.\ a subordering $L_S$.
  Let~$L_S$ be a subordering with free vertices $T$ and let $L_{S^\prime}$ be a \emph{left extension} of $L_S$, that is, $S^\prime\supseteq S$, $L_{S^\prime}[S]=L_S$, and $u\preceq_{L_{S^\prime}}v$ for $u\in S^\prime\setminus S$ and $v\in S$. Then for all $v\in S$ $\wreachright_r(G,L,v)\subseteq \wreachright_r(G,L^\prime,v)$.
\end{lemma}
\begin{proof}
  Let $v\in S$ and let $u\in \wreachright_r(G,L_S,v)$.
  We show that $u\in\linebreak \wreachright_r(G,L_{S^\prime},v)$.
  There are two cases:
  \begin{itemize}
    \item Case 1: $u\in\wreach_r(G[S],L_S,v)$.
          We know that $G[S^\prime]\supset G[S]$ and $L_{S^\prime}[S]=L_S$, hence $u\in \wreach_r(G[S^\prime],L^\prime, v)$.
    \item Case 2: There exists a path $P$ of length at most~$r$ from~$v$ to~$u$ in~$G$ such that $V(P)\cap T=\{u\}$.
          Let $T^\prime=V(G)\setminus S^\prime$ be the free vertices with respect to $L_{S^\prime}$.
          If $u\in T^\prime$, then $V(P)\cap T^\prime=\{u\}$, and $u\in\wreachright_r(G,L_{S^\prime},v)$ because of the same path $P$.
          If $u\in S^\prime$, then notice that $u\in \wreach(G[S^\prime],L_{S^\prime},v)$ because of $P$, as this path certainly does not go left of~$u$ w.r.t.\ $L_{S^\prime}$ as $L_{S^\prime}$ is a left extension of~$L$.
  \end{itemize}
  In both cases $u\in\wreachright_r(G,L^\prime,v)$, which concludes the proof.
\end{proof}

\begin{table}[!t]
  \centering
  \caption{Basic statistics of test graphs.}
  \label{table:test-graphs}
  \begin{tabularx}{\linewidth}{l X X X X | X X X X }
    \toprule
    & \multicolumn{4}{c}{Vertex count} & \multicolumn{4}{c}{Edge count}\\
    \midrule
    Group & min & med & avg & max & min & med & avg & max   \\
    \midrule
    Small & 34 & 115 & 222.52 & 620 & 62 & 612 & 520.61 & 930 \\
    Medium & 235 & 1,302 & 1,448.44 & 4,941 & 1,017 & 3,032 & 3,343.44 & 8,581 \\
    Big & 1,224 & 7,610 & 7,963.64 & 16,264 & 10,445 & 21,000 & 19,519.00 & 47,594 \\
    Huge & 3,656 & 27,775 & 34,598.69 & 77,360 & 48,130 & 186,940 & 237,300.06 & 546,487 \\
    \bottomrule
  \end{tabularx}
\end{table}

\subsection{Further experimental analysis}\label{appendix:furtheranal}
\begin{table}[!t]
  \centering
  \caption{Detailed analysis of the algorithm executions.}
  \label{table:analysis}
  \begin{tabularx}{\linewidth}{l X X X X | X X X X | X X X X }
    \toprule
                            & \multicolumn{4}{c}{\TCLASTC} & \multicolumn{4}{c}{\TCMERGE} & \multicolumn{4}{c}{\TCLASTCRL}                                                          \\\midrule
                            & min                          & max                          & avg                            & med & min & max & avg  & med & min & max  & avg  & med \\\midrule
    $c$ Imp.                & 1                            & 15                           & 1.7                            & 1   & 1   & 9   & 1.2  & 1   & 1   & 14   & 1.6  & 1   \\
    $c$ nImp.               & 1                            & 17                           & 4.7                            & 4   & 1   & 460 & 61.4 & 5   & 1   & 14   & 4.7  & 4   \\
    $t_{TC}$/$t$ Imp.       & 0                            & 1                            & 0                              & 0   & 0   & 1   & 0.4  & 0   & 0   & 1    & 0.5  & 0.5 \\
    $t_{TC}$/$t$ nImp.      & 0                            & 1                            & 0.9                            & 1   & 0.9 & 1   & 1    & 1   & 0   & 1    & 0.9  & 0.9 \\
    \#nodes Imp.            & 1                            & 1e8                          & 157                            & 4   & 2   & 4e7 & 1e4  & 16  & 1   & 1e7  & 168  & 2   \\
    \#nodes nImp.           & 1                            & 1e8                          & 2e3                            & 6   & 2   & 3e7 & 1e5  & 694 & 1   & 3e7  & 1e4  & 10  \\
    $\text{cnt}_{TC}$ Imp.  & 0                            & 200                          & 1.4                            & 0   & 0   & 578 & 3.4  & 0   & 1   & 6327 & 167  & 73  \\
    $\text{cnt}_{TC}$ nImp. & 1                            & 2e4                          & 833                            & 224 & 1   & 624 & 6.5  & 2   & 1   & 8304 & 87.8 & 23  \\
    depth/$c$ Imp.          & 0.1                          & 1                            & 0.3                            & 0.2 & 0.1 & 1   & 0.5  & 0.4 & 0.1 & 1    & 0.4  & 0.5 \\
    depth/$c$ nImp.         & 0.1                          & 1                            & 0.2                            & 0.2 & 0.1 & 1   & 0.3  & 0.3 & 0.1 & 1    & 0.2  & 0.1 \\
    \bottomrule
  \end{tabularx}
\end{table}
\cref{table:analysis} gives detailed descriptions of some measures for executions of the turbocharging algorithms.
The rows come in pairs and describe runs of a turbocharging algorithm or a turbocharged heuristic when turbocharging was successful (Imp.) or when turbocharging was not successful (nImp.).
The table shows the following measures:
\begin{itemize}
  \item $c$: The reconstruction parameter.
  \item $t_{TC}$/$t$: The fraction of time that was used for turbocharging during an application of a turbocharged heuristic.
  \item \#nodes: The number of search tree nodes for an application of a turbocharging algorithm.
  \item $\text{cnt}_{TC}$: The number of times turbocharging was applied during an application of a turbocharged heuristic.
  \item depth/$c$: The ratio of the depth at which a search tree algorithm was cut off and the reconstruction parameter $c$. Note that the maximum depth for the search tree is always $c$.
\end{itemize}
We gathered these values during the executions of our experiments and \cref{table:analysis} reports some statistics on these values broken down by the turbocharging algorithms.

Most notably, we can see that successful turbocharging algorithm applications mostly have small reconstruction parameters $c$, the median being one for all turbocharging approaches.
On the other hand, unsuccessful executions of turbocharging algorithms have larger reconstruction parameters as has to be expected.
It is also interesting that the maximum value of $c$ for \TCMERGE\ is 460, which seems rather high for the depth of a search tree algorithm.
But this is because sets $S_2$ were chosen for merging, such that a lower bound immediately told us that we cannot find an extendable ordering during this application of \TCMERGE, without even having to enter the search tree.

Furthermore, successful runs of turbocharging algorithms explore surprisingly little search tree nodes, as shown by the low average and median values.
Also, \TCLASTCRL\ only requires 2 search tree nodes in most of the cases, which means that it is mostly the case, that a heuristic wants to place a vertex that is ``suboptimal'', while placing nearly any other vertex will achieve lower weak coloring number.
The fraction of time spent on turbocharging is also really low, which means that whenever we can improve the weak coloring number achieved by a heuristic, this can be done easily and in little time.
The same values for unsuccessful runs of turbocharging algorithms look quite different, with far more explored search tree nodes and ratios of time occupied by turbocharging nearly equal to one.

It is also worth mentioning that the median of $\text{cnt}_{TC}$ is zero for \TCLASTC\ and \TCMERGE\ for the successful applications of turbocharged heuristics.
The explanation is that our optimization of the Degree-Heuristic for the left-to-right approaches that places vertices $v$ with $\wreach_r(G,L_s)=k$ immediately to the right of $L_S$, so that an ordering $L$ with $\wreach_r(G,L)=k$ can be computed without even having to turbocharge once.
Of course, this is not always possible.
Furthermore, this also skews the values of reconstruction parameters on successful turbocharging applications.

The values depth/$c$ show that our applied lower bounds often result in less search space that has to be explored due to cutting off of search tree nodes.

\begin{figure}[!t]
  \centering
  \includegraphics[width=\textwidth]{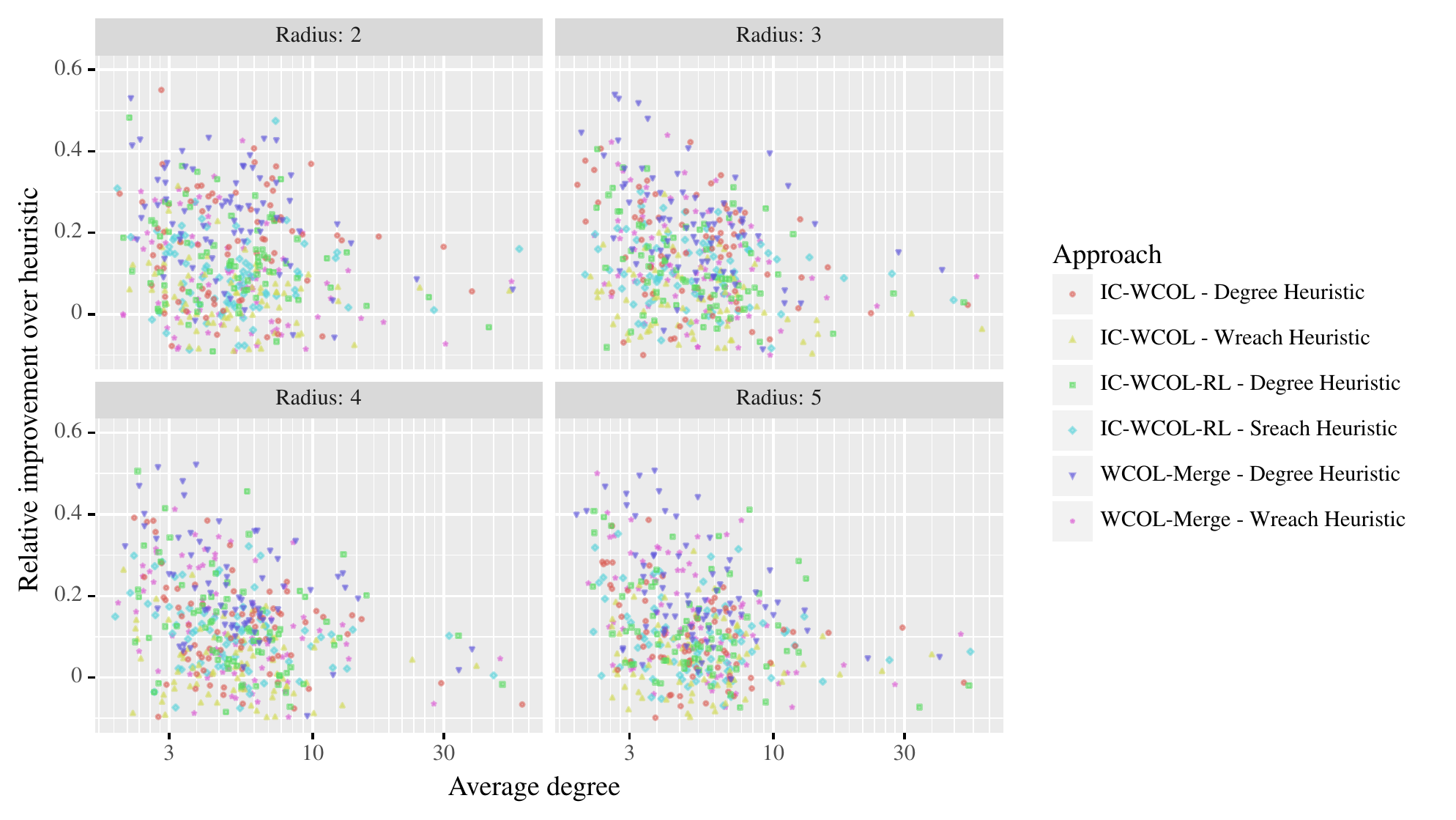}
  \caption{Relative improvement achieved by the turbocharging algorithms when compared to the underlying heuristics subject to the average degree of the input graph broken down by radius.}
  \label{fig:avgdegrelimp}
\end{figure}
\begin{figure}[!t]
  \centering
  \includegraphics[width=\textwidth]{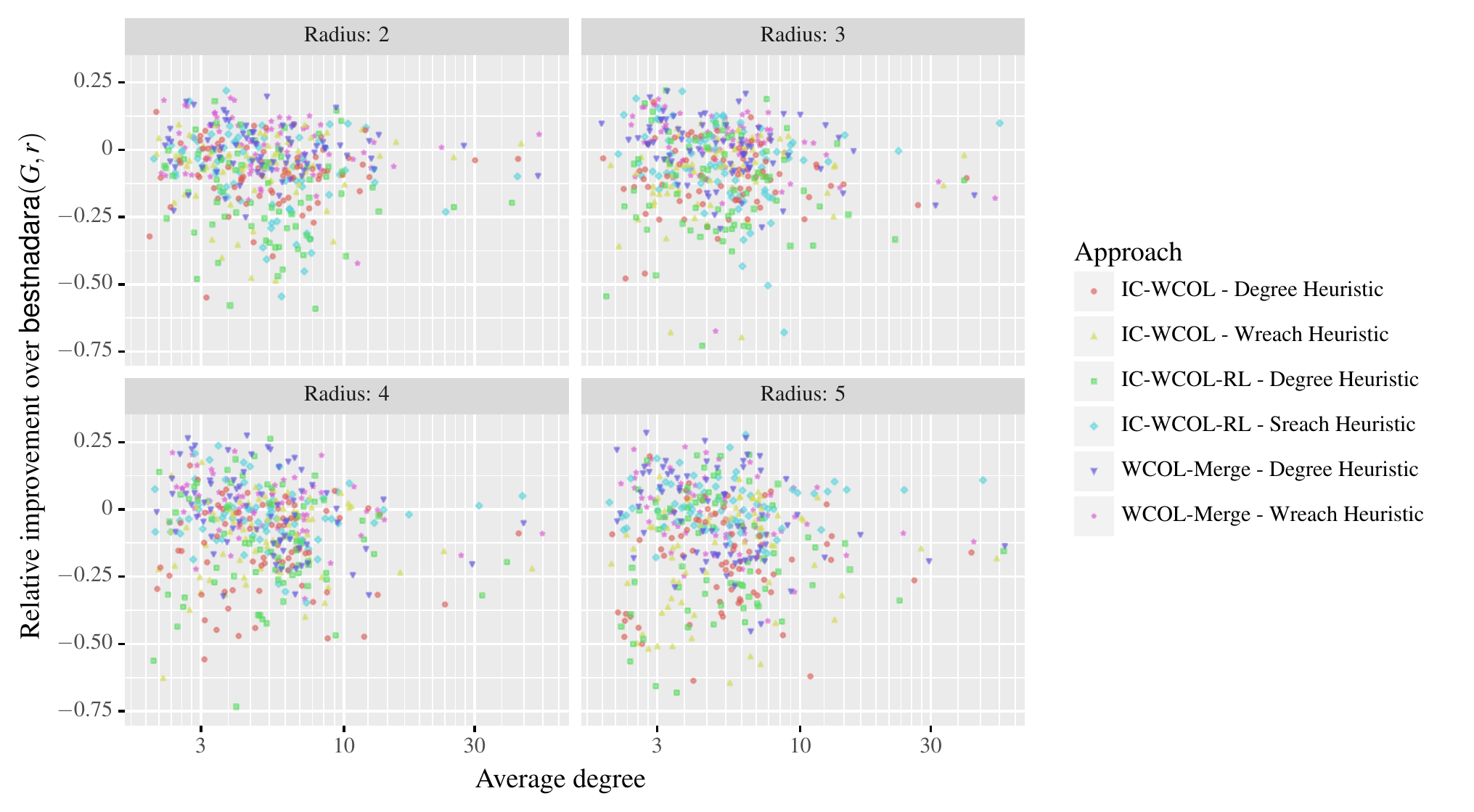}
  \caption{Relative improvement achieved by the turbocharging algorithms when compared to the best algorithm of Nadara et al.\ subject to the average degree of the input graph broken down by radius.}
  \label{fig:avgdegimpnad}
\end{figure}
\subsection{Comparing with a lower bound}\label{sec:lower-bound}
\newcommand{\diam}{\ensuremath{\mathrm{diameter}}}
\newcommand{\degree}{\ensuremath{\mathrm{degree}}}
\looseness=-1
A trivial lower bound for the weak $r$-coloring number of a graph is its degeneracy.
So far, no nontrivial lower bounds for the weak $r$-coloring numbers of Nadara et al.'s dataset were known.
In this section, we propose a heuristic that computes lower bounds based on contracting subgraphs of bounded diameter and compare its results to the best upper bounds.
Our original goal was to use it to speed up the turbocharging algorithms, but this was unsuccessful so far.

Our heuristic is motivated by the MMD+ lower-bounding algorithm for treewidth by Bodlaender and Koster~\cite{DBLP:journals/iandc/BodlaenderK11} that computes the degeneracy of minors of a graph~$G$.
In our case, we cannot contract subgraphs of arbitrary size into single vertices, but we can contract subgraphs of diameter $\lfloor \frac{r-1}{2}\rfloor$, as shown below.
Let us introduce some notation.
A graph $H$ is a \emph{minor} of a graph $G$, if there are pairwise vertex-disjoint connected subgraphs $H_1,\dots ,H_n$ of $G$ such that whenever $\{v_i,v_j\}\in E(H)$, there are $u_i\in V(H_i)$ and $u_j\in V(H_j)$ with $\{u_i,u_j\}\in E(G)$.
We then call $(H_1,\dots ,H_n)$ a \emph{minor model} of $H$ in $G$, and let $\phi(v_i)=H_i$.
Our lower bound is based on the following lemma.
\begin{lemma}%
  \label{lemma:diameterrlowerbound}
  Let $G$ be a graph and let $H$ be a minor of $G$ such that for its minor model $(H_1,\dots ,H_n)$ each $H_i,i\in\{1,\dots, n\}$, has diameter at most $\lfloor\frac{r-1}{2}\rfloor$; then $\degen(H)+1\le \wcol_r(G)$.
\end{lemma}
This bound is similar to Lemma~3.1 in Ref.~\cite{sparsitylecturenotes19}, but instead of requiring $H_i$ to have radius bounded by $(r - 1)/4$ we only require diameter bounded by $\lfloor\frac{r-1}{2}\rfloor$.
We need this stronger result because our heuristic does not guarantee the corresponding bound on the radius.
\begin{proof}[Proof of \cref{lemma:diameterrlowerbound}]
  Let $L$ be an arbitrary vertex ordering of $G$ with $\wcol_r(G,L)=\wcol_r(G)$.
  We construct an ordering $L_{V(H)}$ that witnesses $\degen(H) + 1 \le \wcol_r(G,L)$.
  For each $v\in V(H)$, let $\alpha(v)\in V(\phi(v))$ be the leftmost vertex of $V(\phi(v))$ w.r.t.\ $L$.
  We define~$L_{V(H)}$ such that $v_1\prec_{L_{V(H)}}v_2$ iff.\ $\alpha(v_1)\prec_{L}\alpha(v_2)$.
  We claim that $|\{v^\prime\in N_H(v)\mid v^\prime \prec_{L_{V(H)}}v\}|+1\le |\wreach_r(G,L,\alpha(v))|$ for all $v\in V(H)$.
  To see this, let $v\in V(H)$ be arbitrary and let $v^\prime\in N_H(v)$ s.t.\ $v^\prime\prec_{L_{V(H)}}v$.
  As $v$ and $v^\prime$ are adjacent in $H$, and $\alpha(v)$ and $\alpha(v^\prime)$ have diameter at most $\lfloor\frac{r-1}{2}\rfloor$, the shortest path $P$ between $\alpha(v)$ and $\alpha(v^\prime)$ has length at most~$r$.
  As $v^\prime\prec_{L_{V(H)}}v$ and by definition of $\alpha$, vertex $\alpha(v^\prime)$ is the leftmost vertex on $P$ w.r.t.\ $L$.
  Hence, $\alpha(v^\prime)\in \wreach_r(G,L,\alpha(v))$.
  The claim follows as $\alpha(v^\prime)$ is distinct for each neighbor~$v^\prime$ of $v$ and since $\alpha(v)\in\wreach_r(G,L,\alpha(v))$ accounts for the plus~1.
  Thus, we have that
  \begin{align*}
    \wcol_r(G)= \wcol_r(G,L) & \ge \max_{v\in V(H)}|\wreach_r(G,L,\alpha(v))| \\ &\ge \max_{v\in V(H)}|\{v^\prime\in N_H(v)\mid v^\prime \prec_{L_{V(H)}}v\}|+1\\&\ge \degen(H)+1\text{.}\qedhere
  \end{align*}
\end{proof}

\begin{algorithm}[!t]
  \KwIn{A graph $G=(V,E)$}
  \KwOut{A lower bound for the weak $r$\nobreakdash-coloring number $G$}
  $\textup{\textsf{answer}}\gets 1$;\\
  $H\gets G$;\\
  \While{$V(H)\ne \emptyset$}
  {
    select~$v$ from $V(H)$ that has minimum $\degree_{H}(v)$;\\
    \textup{\textsf{answer}}$\gets\max(\textup{\textsf{answer}},\degree_{H}(v)+1)$;\\
    \uIf{$\exists w\in N_H(v):\diam(G[V(\phi(v))\cup V(\phi(w))])\le \lfloor\frac{r-1}{2}\rfloor$\label{line:mmdplusif}}{
      Contract~$v$ and $w$ in $H$, creating a new vertex~$u$;

    }
    \Else{
      $H\gets H-v$
    }}
  \Return{\textup{\textsf{answer}};}
  \caption{WCOL-MMD+}
  \label[algorithm]{alg:mmdpluswcol} %
\end{algorithm}

We call the lower bound WCOL-MMD+ and its computation is given in \cref{alg:mmdpluswcol}.
The algorithm initializes the graph $H$ to be equal to $G$.
Then it iteratively contracts a vertex $v$ that has minimum degree in $H$ with a neighbor~$w$ into a new vertex $u$ if the model of~$u$ has diameter at most $\lfloor\frac{r-1}{2}\rfloor$.
If no such neighbor $w$ exists, it deletes~$v$.
It returns the largest $\degree_H(v)$ of a vertex $v$ as above that it encounters during the execution.
As $\degree_H(v)\le \degen(H)$ for a vertex $v$ that has minimum degree in $H$, this is indeed a lower bound.
We select that neighbor~$w$, among all choices, that has minimum degree, motivated by a strategy of Bodlaender and Koster.

\begin{table}[t!]
  \centering
  \caption{Average ratios of lower bounds to best known upper bounds.}
  \label{table:ublbcomparison}
  \begin{tabularx}{.75\linewidth}{l c >{\centering\arraybackslash}X >{\centering\arraybackslash}X}
    \toprule
    tests                 & $r$ & $\wcol_1(G)$ & WCOL-MMD+ \\\midrule
    \multirow{4}*{small}  & 2   & 0.512        & 0.512     \\
                          & 3   & 0.390        & 0.433     \\
                          & 4   & 0.328        & 0.363     \\
                          & 5   & 0.306        & 0.361     \\
    \midrule
    \multirow{4}*{medium} & 2   & 0.441        & 0.441     \\
                          & 3   & 0.292        & 0.314     \\
                          & 4   & 0.244        & 0.260     \\
                          & 5   & 0.212        & 0.261     \\
    \midrule
    \multirow{4}*{big}    & 2   & 0.375        & 0.375     \\
                          & 3   & 0.201        & 0.209     \\
                          & 4   & 0.144        & 0.148     \\
                          & 5   & 0.119        & 0.144     \\
    \bottomrule
  \end{tabularx}
\end{table}

\cref{table:ublbcomparison} shows the average ratios between the best knows upper bounds and the two lower bounds $\wcol_1(G)$ and WCOL-MMD+.
The values between both columns do not differ substantially and values are far from 1.
Still, we obtain a modest improvement over the trivial lower bound that is $\wcol_1(G)$.

\end{document}